\documentclass{article}
\usepackage{times}
\usepackage[letterpaper,margin=1in]{geometry}
\usepackage{iftex}
\usepackage{mathpazo}

\usepackage{natbib}
\usepackage{url}

\setcitestyle{authoryear}
\bibpunct[, ]{(}{)}{,}{a}{}{,}%
\usepackage{graphicx}
\usepackage{epsfig}
\usepackage{subfig}
\usepackage{xcolor}

\usepackage{booktabs}
\usepackage{makecell}

\usepackage{amsmath,amsthm,amsfonts}
\usepackage{thmtools} 
\usepackage{bbm}
\usepackage{thm-restate}
\newtheorem{theorem}{Theorem}
\newtheorem{proposition}{Proposition}
\newtheorem{example}{Example}
\newtheorem{lemma}{Lemma}
\newtheorem{corollary}{Corollary}

\usepackage{algorithm}
\usepackage{algorithmic}

\usepackage[
    colorlinks=true,
    linkcolor=blue,
    citecolor=magenta,
    urlcolor=red
]{hyperref}
\usepackage[capitalise,noabbrev]{cleveref}

\usepackage{enumitem}

\newcommand{\btheta}{\boldsymbol{\theta}}
\newcommand{\bthetahat}{\hat{\boldsymbol{\theta}}}

\newcommand{\lambdahat}{\hat{\lambda}}

\newcommand{\sam}{\mathrm{SAM}}
\newcommand{\opt}{\mathrm{OPT}}

\newcommand{\Def}{\operatorname{Def}}
\newcommand{\cbar}{\bar{c}}
\newcommand{\eps}{\varepsilon}

\title{Matching with Choice: Menu Design for Patient--Provider Assignment}

\begin{document}
\author{
\mbox{Naveen Raman\textsuperscript{1}} \thanks{%
  \protect\begin{tabular}[t]{@{}l@{}}
    Correspondence: 
    \texttt{naveenr@andrew.cmu.edu}
  \end{tabular}%
},
\mbox{Holly Wiberg\textsuperscript{1}}
}

\newcommand{\affiliations}{
\mbox{\textsuperscript{1}\,\href{https://cmu.edu}{Carnegie Mellon University}}
}

\makeatletter
\renewcommand{\@maketitle}{%
  {\centering
    {\bfseries\fontsize{18}{22}\selectfont \@title\par}
    \vspace{12pt}
    {\small
      \def\@makefnmark{\hbox{\@textsuperscript{\normalfont\@thefnmark}}}%
      \@author\par}
    \vspace{6pt}
    {\footnotesize\hypersetup{urlcolor=black}\affiliations\par}%
  }
}

\maketitle

\begin{abstract}
    High-quality matches between patients and primary care providers improve downstream health outcomes, yet provider turnover breaks these matches and forces patients to be reassigned. 
Health systems facing this problem must decide not only which providers to recommend, but how much choice to offer. 
Direct assignment avoids competition for scarce provider capacity but relies on noisy estimates of patient preferences; broader choice allows patients to correct these estimates using private information but can create costly coordination failures.
We study this tradeoff through the design of fixed, personalized menus from which patients choose providers.
We show that constructing the optimal menu is NP-hard and establish an information-theoretic lower bound on performance. 
We find that offering each patient a single provider matches this bound up to a constant factor, while unrestricted choice can perform arbitrarily poorly. 
However, we show that additional choice improves on single-provider assignment when the system cannot reliably distinguish among a patient's alternatives. 
We formalize this tradeoff through two forms of loss: menu loss, which arises when menus offer too few options, and coordination loss, which arises when scarce capacity is exposed to too much choice. 
Motivated by this decomposition, we develop the Scenario-Averaged Marginals (SAM) algorithm, which tailors menus based on provider shadow prices so choice is offered only when utility exceeds the capacity cost. 
We identify sufficient conditions under which SAM's regret is of the same order as the lower bound, namely settings where patients have valuable alternatives and offered menus do not create competition for scarce capacity.
In a semi-synthetic environment calibrated to a provider network using Medicare data, SAM improves total utility by up to 19\% over single-provider assignment as preference uncertainty grows, while matching it when estimates are accurate.
Beyond aggregate welfare, we find tailored choice improves utility for the least-served, largely rural ZIP codes relative to direct assignment, without an explicit fairness objective. 
Our results demonstrate that patient choice can serve as an operational tool for mitigating preference uncertainty, but that realizing its benefits requires balancing it against the coordination of scarce provider capacity.

\end{abstract}
\section{Introduction}
\label{sec:intro}
Primary care providers (PCPs) are essential to the healthcare ecosystem because they are the first point of contact for many patients~\citep{impact_primary_care}.
The benefits of PCPs compound with continuity of care, which is associated with better clinical outcomes~\citep{continuity_readmission} and higher provider productivity~\citep{continuity_productivity}. Patients are more likely to stay with providers they trust and have chosen~\citep{continuity_switching, choice_continuity}, making high-quality matches between patients and providers an important system objective. 
Unfortunately, high provider turnover rates frequently lead patients to lose their PCP, which disrupts patient care and leads to worse health outcomes~\citep{pcp_turnover}. 
Each of these patients must be rematched, yet the process can take months due to provider scarcity and the administrative burden of patient coordination~\citep{finding_new_provider}. 

During the rematching process, health systems must decide which patients to match with which available providers, with the goal of maximizing match quality across this population. Two common approaches lie at the extremes: direct assignment, where the system offers one provider per patient based on its own estimate of fit, and full choice, where patients select among all providers with no central coordination. Beyond its documented benefits for patient satisfaction and retention~\citep{choice_continuity}, patient choice is an informational lever: it can reveal preferences unavailable to the system, allowing patients to select their best match rather than relying on the system's noisy estimates. Yet choice comes at an operational cost, since when many patients are given access to the same providers, early selections can deplete scarce capacity. Thus, health systems face a key design decision: how much choice to offer, and which providers to offer each patient. 

\begin{figure*}
    \includegraphics[width=\textwidth]{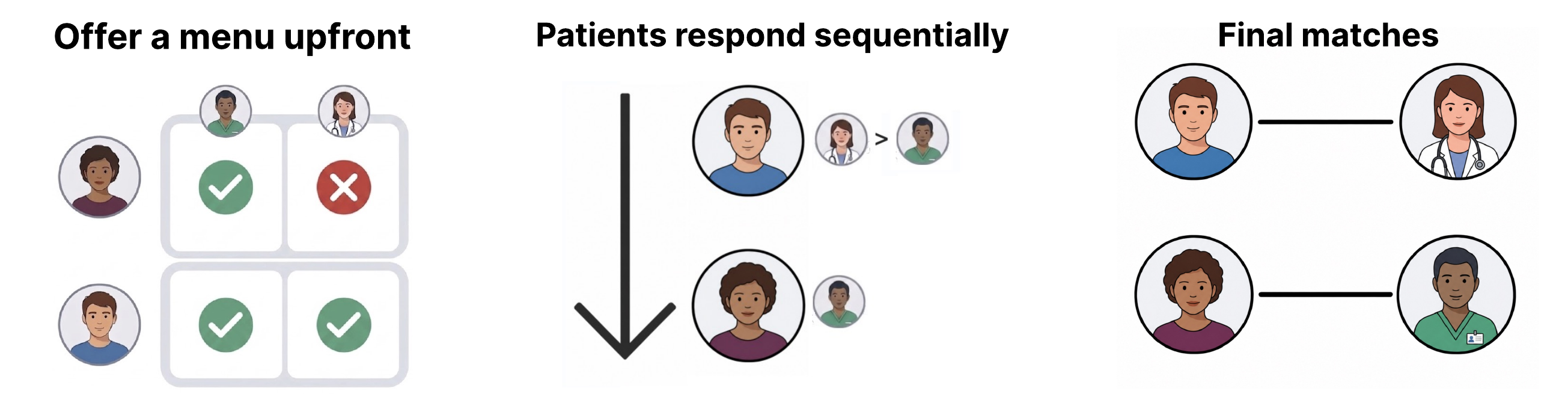}
    \caption{We study when to offer choice in patient--provider matching. The system first offers each patient a personalized menu of providers. Patients then respond sequentially in random order and select among available providers according to their true, unobserved preferences, resulting in a final matching.}
    \label{fig:pull}
\end{figure*}

Rematching is a natural setting in which to study this design problem. A provider departure displaces a known group of patients at once, who must be absorbed by a known set of providers with limited capacity. Additionally, unlike for new patients, the health system has prior history from which to estimate these patients' preferences. We therefore consider the setting where the health system initially offers each displaced patient a personalized set of options, 
and patients then respond sequentially to select providers from this set, subject to capacity constraints (Figure~\ref{fig:pull}).
Our objective is to design menus that maximize the average utility per patient. 

We first study the fundamental limits of menu design. 
We show that finding the optimal menu is computationally intractable, with an information-theoretic lower bound on the performance of any algorithm. 
We establish that a simple ``Offer-One'' policy, which offers at most one provider to each patient, matches this lower bound up to a constant factor, and is therefore minimax-optimal among deterministic algorithms. 
Yet Offer-One is not optimal for every market. 
When there is uncertainty in patient preferences, additional patient choice allows patients to act on preference information unavailable to the system. 
However, the opposite extreme, where all options are offered (``Offer-All''), can incur arbitrarily poor performance. 
Thus, neither extreme captures the full value of patient choice: Offer-One is robust but can forgo valuable patient information, while Offer-All can exploit that information at the cost of severe capacity competition.

These results motivate the question of when we can use patient choice to improve upon Offer-One without incurring the coordination failures of Offer-All. 
We decompose the gap between any menu and the omniscient upper bound into two terms: menu loss and coordination loss. 
The menu loss term captures the quality of the offered options relative to the optimal matching if true preferences were known.
The coordination loss captures loss due to the sequential arrival of patients, and in particular, captures how much sequential depletion reduces performance. 
Expanding menus can reduce menu loss by giving patients access to more valuable alternatives, but can increase coordination loss by exposing scarce providers to more patients. 
Thus, menu design requires balancing the informational value of patient choice against the operational value of capacity coordination.

Guided by this tradeoff, we develop the Scenario-Averaged Marginals (SAM) algorithm, which tailors menus based on provider shadow prices to balance patient--provider utility with provider capacity. 
SAM samples plausible preference matrices and averages the difference between utilities and shadow prices, thereby capturing the net value of an option. 
This produces menus that are less sensitive to any single realization of uncertain patient preferences.
We derive general performance guarantees and give sufficient conditions under which SAM achieves regret of the same order as the lower bound. 
These conditions capture markets in which patients have sufficiently good alternative providers and additional choice does not create competition for scarce capacity; we return to settings with capacity competition in our empirics. 
We illustrate cases where SAM strictly outperforms Offer-One, demonstrating that richer menus can yield gains beyond the worst-case guarantee. 
At the same time, we show that no policy uniformly dominates; Offer-One performs better under perfect information, while Offer-All performs well when inter-patient competition is weak. 

We then compare these policies in a semi-synthetic environment based on real Medicare data. 
We characterize the value and structure of patient choice as a function of market conditions.
As preference estimates become noisier, richer menus become more valuable because patients can correct consequential estimation errors.
As provider capacity increases, coordination costs fall, making broader choice less costly and allowing untailored menus to approach the performance of tailored ones.
In our semi-synthetic environment, tailored menus also improve the quality of options presented to patients and reduce geographic disparities in match quality, with particularly large gains for rural patients. 

Our contributions are as follows: 
\begin{enumerate}
    \item We introduce a menu-based model of patient--provider matching that captures capacity constraints, preference uncertainty, and sequential choice from menus committed offline.
    \item We characterize the fundamental limits of menu design. We establish computational and information-theoretic lower bounds, show that Offer-One is minimax rate-optimal among deterministic algorithms, and develop a decomposition separating menu loss from coordination loss.
    \item We develop Scenario-Averaged Marginals (SAM), a shadow-price-based algorithm for constructing personalized menus. We provide sufficient structural conditions under which SAM achieves regret of the same order as the lower bound. 
    \item We empirically characterize when patient choice is valuable in a semi-synthetic environment calibrated to Medicare data, showing how preference uncertainty, provider capacity, and preference structure shape the performance and composition of personalized menus.
\end{enumerate}
\section{Literature Review}
\label{sec:literature_review}
Our work relates to three complementary strands of literature: work on patient--provider matching, assortment optimization under choice and uncertainty, and matching markets with capacity-constrained and online assignment.

\subsection{Patient--Provider Matching}
\label{subsec:related_patient_provider}
The clinical literature demonstrates that determining what constitutes a ``high-quality'' match is challenging. 
Empirical evidence suggests that geographic proximity is an important factor in patient preferences~\citep{patient_decisions}, with concordance in race~\citep{race_concordance}, gender~\citep{gender_concordance}, and language~\citep{language_concordance} playing a secondary role. 
Concordance along these dimensions can build trust, which can translate to better care~\citep{race_concordance_lung}. 
As one illustration, \citet{concordance_race_cardio} suggest that increasing access to race-concordant physicians could reduce Black-White disparities in cardiovascular mortality among men by 19\%.
Beyond concordance, wait times and provider schedules impact patient utilities~\citep{waiting_see_doctor}. 
This literature provides valuable insight into patient preferences, yet how health systems should operationalize matching under capacity constraints and incomplete information remains an open question.

From an algorithmic perspective, prior work on patient–provider matching typically assumes limited or no patient choice and adopts either one-shot assignment or heuristic optimization frameworks. Within the one-shot framework, prior work has developed linear programs to manage provider panel sizes~\citep{panel_size_lp}, deferred acceptance frameworks to perform two-stage matching~\citep{two_stage_matching}, and scheduling algorithms to model optimal rates of on-demand scheduling~\citep{revenue_management_patient_choice}. 
A separate line of work uses genetic algorithms to find fair matches~\citep{chen_genetic} or balance provider workloads~\citep{genetic_non_dominated}. 
Related work has also studied physician recommendation algorithms that account for patient choice and limited provider capacity~\citep{physician_recommendation}.

In contrast, we construct patient menus upfront, before preferences are revealed, and study how choice can mitigate preference uncertainty while introducing competition for scarce provider capacity.
Patient choice is a central consideration in healthcare delivery, as it supports patient autonomy and shared decision-making~\citep{importance_patient_choice,value_of_choice}, and prior work shows that involving patients in provider selection and treatment decisions can improve both satisfaction and system performance~\citep{informed_patient,choice_national_healthcare, choice_continuity}.
In this work, we do not aim to model the full behavioral process underlying patient decision-making; instead, we treat patient choice as both an institutional constraint and an algorithmic mechanism for hedging against uncertainty in estimated patient–provider match quality. 

\subsection{Assortment Optimization}
\label{subsec:related_assortment}
Our setup shares commonalities with assortment optimization, in which a platform offers curated option sets and agents select from them~\citep{assortment_optimization}. 
Retailers use models to capture customer choice behavior, such as multinomial logit~\citep{assortment_mnl} and nested logit~\citep{assortment_nested_logit}. Our setting combines three features that distinguish it from classical assortment optimization and together create a tradeoff between information and coordination.
First, we construct menus for a known population of patients whose choices are coupled through shared provider capacities.
Second, menus are constructed offline but exercised sequentially in an unknown response order, so individually attractive assortments can generate capacity externalities across patients.
Third, uncertainty concerns patient-specific match utilities: each patient has a realized utility for each provider, but the platform observes only a noisy estimate of that utility.

Our work differs from both the robust assortment literature and the capacity-constrained assortment literature. 
In contrast to robust assortment optimization~\citep{robust_assortment}, our uncertainty concerns patient-specific match utilities that patients themselves can act on when choosing from their menus.
Capacity-constrained assortment optimization~\citep{real_time_personalized,dynamic_assortment,assortment_planning_constraints} decides what to offer customers as they arrive and inventory depletes. In healthcare, ~\citet{appointment_booking} study which appointment slots to offer arriving patients whose slot preferences are unknown, with the goal of maximizing the number of slots booked. In contrast to this literature, our patient population is known when menus are constructed, menus are committed upfront, and the objective is match quality.
Rather than hedging against who will arrive, we hedge against patient ordering. 

\subsection{Choice and Congestion in Matching Markets}
\label{subsec:related_matching}
By giving patients choice over their assignment, our work connects to the matching markets literature. 
Classical matching models, such as stable matching~\citep{stable_matching}, with applications such as school choice~\citep{school_choice}, use reported preferences to produce stable assignments; we instead leave the choice directly to the patient. 
Related work has studied matching markets in which agents face limits on the number of preferences they can report, and how these constraints affect equilibrium outcomes~\citep{constrained_school_choice}, the closest analogue to a menu size limit. 

Our setting also connects to online and dynamic matching due to the sequential response of patients. In bipartite online matching, optimal algorithms achieve a $1-\frac{1}{e}$ competitive ratio~\citep{optimal_bipartite}. 
~\citet{batch_greedy_assortment} study batching and greedy policies in dynamic matching and show that simple policies become asymptotically optimal as agents can remain in the market longer.
Similar to our setting is~\citet{online_matching_frameworks}, who study online matching where customers arrive sequentially and select the first acceptable product from a sequence of options presented by the platform. In contrast, our population is fixed and we construct fixed menus upfront under imperfectly observed utilities.

Closest to our work is the literature on assortment planning for matching platforms, which confronts a tension between choice and capacity. Such menu design problems arise in dating platforms~\citep{assortment_dating}, school choice~\citep{priority_based_allocation,high_school_match}, and ride-hailing~\citep{driving_assortment}. 
Our setting differs in two ways. First, these models generally receive user selections as requests that the platform centrally resolves, whereas in our setting patient selection immediately depletes capacity, so the menu is the only coordination mechanism. Second, choice is typically random with parameters known by the platform, whereas in our case patients hold preference information the platform lacks and are resolved through menu offers. 
Our work also differs from \citet{assortment_planning_two_sided} and \citet{assortment_online}, who study menu design under two-sided sequential choice, focusing on maximizing the expected number of matches rather than welfare under uncertain preferences.
\section{Patient--Provider Matching Model}
\label{sec:model}
This section introduces the patient--provider matching model. 
We formalize patient utilities, preference uncertainty, capacity constraints, and the offline–online structure of menu offering and response. 

\subsection{Model Primitives}
\label{subsec:patient_model}
We begin by describing our model of patients and providers, including the notion of utility and noise. 
\paragraph{Patients and Providers.}
We consider a setting with \(N\) patients and \(M\) providers.
Let \(\theta_{ij} \in [0,1]\) denote the utility that patient \(i \in [N]\) derives from provider \(j \in [M]\). This is a measure of match quality that can reflect any attributes fixed before menus are offered, such as distance, clinical fit, demographic concordance, and schedule compatibility. 
Each patient also has an \emph{exit option} denoted by an additional column \(j = M+1\) with utility \(\theta_{i, M+1}\), which represents the utility derived from remaining unmatched (e.g., continuing with walk-in service) or waiting for a future matching opportunity.
Let $\Delta_{ij} := \theta_{ij}-\theta_{i,M+1}$ and $\hat{\Delta}_{ij} := \hat{\theta}_{ij}-\hat{\theta}_{i,M+1}$ denote the true and estimated utility advantages, respectively, of provider $j$ over the exit option for patient $i$.

The health system, which we refer to as the platform, does not observe \(\boldsymbol{\theta}\), but rather an estimate \(\hat{\boldsymbol{\theta}}\).  This can be obtained from predictive models that map patient and provider features to satisfaction or other match-quality proxies, or from direct elicitation such as patient surveys; we discuss this estimation task further in Section~\ref{sec:discussion}. 
We assume this estimate approximates the true utilities \(\boldsymbol{\theta}\) up to an error level $\eps$:
\(\boldsymbol{\theta} = \hat{\boldsymbol{\theta}} + \boldsymbol{\delta}\), with
\(\|\boldsymbol{\delta}\|_{\infty} \le \eps\), so that \(\boldsymbol{\theta}\) lies in the \(\eps\)-ball centered at \(\hat{\boldsymbol{\theta}}\):
\(
\mathcal{B} = \{\boldsymbol{\theta} \in [0,1]^{N\times(M+1)} :
\|\boldsymbol{\theta}-\hat{\boldsymbol{\theta}}\|_\infty \le \eps\}
\).
We treat the true utilities as drawn from a distribution $\mathcal P_{\boldsymbol\theta}$ supported on $\mathcal B$ and assume that the perturbation is mean-zero, $\mathbb E[\boldsymbol\delta]=0$. 
$\boldsymbol{\delta}$ represents match-relevant information available to the patient but not captured by the platform's preference estimate, such as preferences over provider characteristics that are known to the patient but imperfectly represented in administrative data. 
Thus, $\boldsymbol{\delta}$ represents a gap in the platform's information, not randomness in a patient's choice conditional on their true utilities. Patients need not know their numerical utilities, only which available provider they prefer. 
 
\paragraph{Selection and Menus.}
We offer a menu \(\mathbf{X} \in \{0,1\}^{N \times M}\), where \(X_{ij} = 1\) indicates that provider \(j\) is offered to patient \(i\).  
We implicitly offer the exit option $j=M+1$ to all patients. 
For each patient \(i\), let $L_{i}$ denote the number of offered providers on the initial menu: 
\[
L_{i} = \sum_{j=1}^{M} X_{ij}.
\]
A larger $L_{i}$ corresponds to more options offered to a patient; that is, patients have more providers to select from. 
We enforce an upper bound $L_{i} \leq k$, where $k \in [M]$ is a hyperparameter. 
By doing so, we limit the number of providers shown to each patient, which captures logistical constraints (e.g., number of providers on the first page of a patient portal) and can also mitigate option overload~\citep{choice_overload}. Pairs that are ineligible, for instance due to insurance network or other hard constraints, can be excluded by fixing \(X_{ij}=0\).

\paragraph{Response Order and Selection.}
Providers have limited capacities represented by the vector \(\mathbf{c}_t \in \mathbb{N}^{M+1}\) at time \(t\), where \(c_{t,j}\) is the remaining capacity of provider \(j\), and \(c_{t, M+1} = N\) represents the exit option. 
Providers initially have capacity $\bar{c}_{j}$, which we assume is known, with total capacity $C:=\sum_{j=1}^{M} \bar{c}_j$. 
Patients respond in some order \(\boldsymbol{\sigma}\), which is a permutation of \([N]\) drawn from a distribution \(\mathcal{D}\): $
\boldsymbol{\sigma} \sim \mathcal{D}$. We assume $\boldsymbol{\sigma}$ and $\boldsymbol{\theta}$ are independent. 
At each step $t$, patient $i=\sigma_t$ selects
\[
s_i
=
\arg\max_{j\in[M+1]}
\left\{
\theta_{ij}
:
j=M+1
\ \text{or}\
\bigl(X_{ij}=1 \text{ and } c_{t,j}>0\bigr)
\right\},
\]
with ties broken by a fixed deterministic rule.
If $s_i\in[M]$, then
\[
c_{t+1,s_i}=c_{t,s_i}-1,
\]
while all other provider capacities remain unchanged.
That is, each patient selects the best available provider at each time step according to their true utilities; the best offered provider might not be available when it comes time to match. 
We assume that, conditional on the menu and remaining capacities, patient selections are deterministic, though we return to this in Section~\ref{sec:empirical}.

\paragraph{Optimization Problem.}
Our goal is to design menus \(\mathbf{X}\) that maximize expected total utility across patients, accounting for preference uncertainty, stochastic response order, and capacity constraints.
We focus on match quality to capture the idea that we not only want to match the largest number of patients but also take into account patient preferences. 
Formally, for a fixed menu \(\mathbf{X}\), patients respond sequentially according to a random order \(\boldsymbol{\sigma} \sim \mathcal{D}\), selecting their most preferred available provider from their offered menu.

For a realized utility matrix $\boldsymbol{\theta}$ and response order
$\boldsymbol{\sigma}$, define the total realized utility under menu
$\mathbf{X}$ as
\[
V_{\boldsymbol{\theta}}(\mathbf{X},\boldsymbol{\sigma})
=
\sum_{i=1}^{N}\theta_{i,s_i(\mathbf{X},\boldsymbol{\sigma};
\boldsymbol{\theta})}.
\]
The expected value of a menu is
\[
V(\mathbf{X})
=
\mathbb{E}_{\boldsymbol{\theta}\sim\mathcal{P}_{\boldsymbol{\theta}},
\,\boldsymbol{\sigma}\sim\mathcal{D}}
\left[
V_{\boldsymbol{\theta}}(\mathbf{X},\boldsymbol{\sigma})
\right].
\]
The platform's menu-design problem is therefore
\[
\max_{\mathbf{X}\in\{0,1\}^{N\times M}}
\quad
V(\mathbf{X})
\qquad
\text{s.t.}
\qquad
\sum_{j=1}^{M}X_{ij}\le k
\quad\forall i\in[N].
\]

\subsection{Offline Offering and Online Response}
\label{subsec:offline_online}
In our setting, menus are constructed upfront and offered to all patients simultaneously, after which patients respond sequentially; the only updates thereafter are capacity-driven reductions as providers fill up. This offline-commit, online-response structure is motivated by two considerations. The first is administrative tractability: menus can be computed in a single batch at the natural decision point (e.g., following a provider departure), reviewed before release, and communicated once through existing channels such as portal messages or mailed letters. A fully online scheme would instead require re-solving and re-issuing tailored offers at every arrival. 

The second is incentive alignment. When offers vary with response time, delayed participation can be rewarded in anticipation of better future options, further lengthening match times. Under offline commitment, the only post-offer change is capacity depletion, so a patient's options can only weakly worsen over time and no patient gains by waiting. Formally, we demonstrate that the expected utility for any patient (with fixed menu $\mathbf{X}$) weakly decreases with later response positions: 
\begin{proposition}[Offline Menus Incentivize Early Response]
\label{prop:early-response}
Suppose the menu \(\mathbf{X}\) is computed offline and patients respond in a random order \(\boldsymbol{\sigma} \sim \mathcal{D}\), with ties broken by a fixed deterministic rule.  
For any fixed ordering $\sigma$ and patient $i$, let $U_i(t; \sigma, \mathbf{X})$ denote the utility patient $i$ obtains when responding at position $t$, holding the relative order of the other patients fixed. 
Then $U_i$ is nonincreasing in $t$ for every $\boldsymbol{\sigma}$, and hence in expectation over $\boldsymbol{\sigma} \sim \mathcal{D}$:
\[
\mathbb{E}_{\boldsymbol{\sigma}}[U_i(t;\boldsymbol{\sigma},\mathbf{X})] \ge 
\mathbb{E}_{\boldsymbol{\sigma}}[U_i(t';\boldsymbol{\sigma},\mathbf{X})]
\quad \forall\, t < t'.
\]
\end{proposition}
We include a proof in Appendix~\ref{app:early_response_proof}, where we show that the monotonicity of provider availabilities in time implies later response positions face a (weakly) smaller available set of providers and hence lower utility.
Notably, this incentive requires no knowledge of other patients' preferences or arrival behavior: it suffices that patients know their menu will not be expanded, a property the platform can make salient by presenting each patient a fixed list with availability indicators. We discuss how repeated rematching rounds preserve this property in Section~\ref{sec:discussion}.

\subsection{Epistemic Uncertainty and Stochastic Choice}
\label{subsec:stochastic_choice}
We briefly discuss the difference between our model and the random utility model literature. 
Random utility models in assortment optimization focus on maximizing expected revenue under behavioral uncertainty, typically making specific assumptions on the distribution of utility shocks (e.g., Gumbel in MNL)~\citep{assortment_mnl}. 
In our baseline model, by contrast, each patient has a fixed realized utility
$\theta_{ij}$ for each provider, but the platform observes only the estimate
$\hat{\theta}_{ij}$.
Uncertainty lies in the platform's information about preferences rather
than in patient behavior conditional on those preferences.
As a result, we bound deviations ($|\delta_{ij}| \leq \eps$) directly and give performance guarantees over this noise, rather than assuming a specific utility-shock distribution. 
We note that we use deterministic selection primarily for proof tractability; the algorithms from Section~\ref{sec:sam} apply unchanged under stochastic choice, and Section~\ref{sec:empirical} shows that the performance advantage persists under an MNL choice model. 
\section{Limits and Value of Patient Choice}
\label{sec:complexity}
We establish two limits on optimal menu design: we prove that 1) finding the optimal menu is computationally intractable and 2) preference uncertainty imposes an information-theoretic limit on performance.
We then consider two extreme policies: offering no choice or offering unrestricted choice. 
We find that a no choice approach is minimax rate-optimal among deterministic algorithms. 
Yet this policy is brittle to uncertainty, and we establish that choice can improve upon it in certain market conditions. On the other hand, offering too much choice leads to a different failure mode of capacity contention. 
These joint results naturally lead to the question of \textit{when} offering choice can improve performance, which we answer in Section~\ref{sec:sam}. 

\subsection{Fundamental Limits of Menu Design}
\label{subsec:fundamental_limits}
While the menu design problem can be solved optimally in limited scenarios, we next demonstrate that a general solution is intractable even when stripped down to its simplest form. 

\begin{theorem}[NP-Hardness of Menu Optimization]
\label{thm:nphard}
The menu optimization problem defined in Section~\ref{sec:model} is NP-hard.
This holds even with only two providers ($M=2$) with unit capacity per provider ($\bar{c}_{1} = \bar{c}_2 = 1$) and a maximum menu size of $k=1$. 
\end{theorem}
We prove this via reduction from the Product Partition problem (full proof in Appendix~\ref{app:intractability_proof}).
Given weights $w_1,\ldots,w_N$, we construct a menu design instance with $M=2$, $k=1$, unit capacities, uniform exit utility $\gamma \in (0,1)$, and Bernoulli utilities $\theta_{ij} \sim \mathrm{Bernoulli}(q_i)$ where $q_i = 1-1/w_i$, with unbiased estimate $\hat\theta_{ij}=q_i$.
For any partition $(S_1,S_2)$, the expected total utility equals $(1-\gamma)[2-1/W_1-1/W_2]+N\gamma$ where $W_j=\prod_{i\in S_j}w_i$.
This is maximized exactly when the partition is balanced, reducing menu optimization to Product Partition. 

We next establish an information-theoretic lower bound. Let $V_{\boldsymbol{\theta}}(\mathbf{X}) = \mathbb{E}_{\boldsymbol{\sigma}}[V_{\boldsymbol{\theta}}(\mathbf{X},\boldsymbol{\sigma})]$ and let $\opt(\boldsymbol\theta) = \max_{\mathbf{X}} V_{\boldsymbol{\theta}}(\mathbf{X})$ (which we shorten to $\opt$) denote the welfare of an omniscient benchmark with full knowledge of the true $\boldsymbol{\theta}$ but not of the ordering $\boldsymbol{\sigma}$. 
The following result shows that, against this benchmark, every algorithm must incur expected regret that scales linearly with $\eps C$ in over-subscribed markets. 
Since $\opt$ is the same for all algorithms, it provides a common standard for comparing performance. 

\begin{theorem}[Information-Theoretic Lower Bound]
    \label{thm:info_theory_lower_bound}
    Let $\mathcal{D}$ be the uniform permutation on $[N]$. 
    For any $\eps < \frac{1}{2}$ with
$N = \alpha C$ for some constant $\alpha > 1$, there exists $\hat{\boldsymbol{\theta}}$
    and a distribution $P_{\boldsymbol{\theta}}$ over $\boldsymbol{\theta}$, such
    that for any deterministic algorithm $\mathcal{A} : \hat{\boldsymbol{\theta}} \mapsto \mathbf{X}$,
    \[
        \mathbb{E}_{\boldsymbol{\theta} \sim P_{\boldsymbol{\theta}}}\Big[\opt - V_{\boldsymbol{\theta}}\!\left(\mathcal{A}(\hat{\boldsymbol{\theta}})\right)\Big]
        \;\geq\;  (\min \{1,\alpha-1\}\eps - o(1)) C.
    \]
where the $o(1)$ term vanishes as $\min_j \bar c_j \to \infty$.
\end{theorem}

We prove this in Appendix~\ref{app:info_theoretic_proof} by constructing an instance where $\alpha \bar{c}_{j}$ patients are competing for $\bar{c}_{j}$ slots for each provider $j$. 
Here, $\hat{\theta}_{ij} = \frac{1}{2}$, while $\delta_{ij} \in \{-\eps, \eps\}$ with equal probability. 
Then the omniscient algorithm uses knowledge of the true $\delta_{i,j}$ to select the best subset of providers.

\subsection{Offer-One: A Minimax Rate-Optimal Solution}
\label{subsec:heuristics}
We next study the Offer-One policy, which offers no choice to patients.
The Offer-One policy relies on the estimated $\hat{\boldsymbol{\theta}}$ and solves a bipartite matching problem, resulting in menus of size 0 or 1. 
\[
\begin{aligned}
\max_{\mathbf{X}} \quad
& \sum_{i=1}^N \sum_{j=1}^M \hat{\Delta}_{ij} X_{ij} \\
\text{s.t.} \quad
& \sum_{j=1}^M X_{ij} \leq 1,
&& \forall i \in [N], \\
& \sum_{i=1}^N X_{ij} \leq \bar c_j,
&& \forall j \in [M], \\
& X_{ij} \in \{0,1\},
&& \forall i \in [N],\ j \in [M].
\end{aligned}
\]
In a noiseless setting ($\eps=0$), Offer-One's menus are optimal since $\btheta = \bthetahat$. However, as noise increases, Offer-One is unable to hedge against this estimation error. 

We next show that the policy matches the lower bound up to constants, and is therefore minimax-optimal among deterministic algorithms. 

\begin{proposition}[Offer-One Performance Characterization]
\label{prop:robustness_main}
Let $\mathbf{X}^{\mathrm{One}}$ be the menu computed from the Offer-One algorithm. 
Then for any $\eps < \frac{1}{2}$, if $N>C$, there exists a $\boldsymbol{\hat{\theta}}$, capacities $\bar{c}_{j}$, and a distribution of $\boldsymbol{\theta}$ such that 
\[
\mathbb{E}[\opt - V_{\boldsymbol{\theta}}(\mathbf{X}^{\mathrm{One}})] \geq \left(\eps-o(1)\right)C.
\]
Here, the $o(1)$ term vanishes as $C\to\infty$.
Moreover, this bound is tight up to constant factors; for any problem instance we have that 
\[
\mathbb{E}[\opt - V_{\boldsymbol{\theta}}(\mathbf{X}^{\mathrm{One}})] \leq 2\eps \min(C,N).
\]
\end{proposition}
We prove this by constructing an instance where noise systematically inverts the ordering between optimal and suboptimal pairs (proof in Appendix~\ref{app:baseline_proofs}). 
This result shows that, in the worst case, choice provides no improvement in the achievable regret rate. A simple matching policy is robust, yet it directly inherits any estimation error that could otherwise be resolved via patient choice, which we turn to next.

\subsection{The Value and Cost of Choice}
\label{subsec:value_of_choice}
Despite the worst-case guarantee on Offer-One, offering choice can improve performance in certain instances by taking advantage of preference information available only to patients. 
The key idea is that when patients have multiple options with similar $\hat{\theta}_{ij}$, offering choice allows patients to select using their realized preferences.  
This becomes more important with larger $\eps$, as potential differences in the true preferences between patients become larger. 
We demonstrate such an example below. 
\begin{example}[Patient Choice Can Improve upon Offer-One]
\label{ex:value-of-choice}
Consider a single patient and two unit-capacity providers.
Suppose the platform estimates both providers to be equally valuable,
\[
\hat{\theta}_{1,1}=\hat{\theta}_{1,2}=\frac{1}{2},
\]
while the patient's realized utilities are
\[
\theta_{1,j}=\hat{\theta}_{1,j}+\delta_{1,j},
\qquad
\delta_{1,j}\overset{\mathrm{i.i.d.}}{\sim}\mathrm{Unif}[-\eps,\eps].
\]

Offer-One must commit to a single provider based on the estimated utilities.
Without loss of generality, suppose it offers provider $1$.
Since the noise is mean zero,
\[
\mathbb{E}[V_{\boldsymbol{\theta}}(\mathbf{X}^{\mathrm{One}})]
=
\mathbb{E}[\theta_{1,1}]
=
\frac{1}{2}.
\]

Now consider offering both providers, with $\mathbf{X} = \mathbf{1}$.
The patient observes their realized preferences and selects the provider with higher realized utility, giving
\[
\mathbb{E}[V_{\boldsymbol{\theta}}(\mathbf{X}^{\mathrm{All}})]
=
\mathbb{E}\left[\max\{\theta_{1,1},\theta_{1,2}\}\right]
=
\frac{1}{2}
+
\mathbb{E}\left[\max\{\delta_{1,1},\delta_{1,2}\}\right]
=
\frac{1}{2}+\frac{\eps}{3}.
\]

Offering choice improves expected utility by $\eps/3$.
Although the two providers are indistinguishable to the platform, the patient can use their realized preference information to select the better provider.
This illustrates the informational value of choice: allowing patients to choose can outperform committing to a single provider based only on estimated utilities.
\end{example}

If choice can reveal information, then another natural policy would be to offer all possible providers to every patient. 
This policy, which we call Offer-All, lets $\mathbf{X} = \mathbf{1}$, and is equivalent to allowing patients free access to a central directory.
While Offer-All maximizes patients' ability to act on their private preferences, it ignores inter-patient competition for scarce providers. 
\begin{proposition}[Offer-All Performance Characterization]
    \label{prop:offer_all_poor}
    For any $\alpha \in (0,1)$, there exists a menu-design instance with preferences $\boldsymbol{\theta}$ such that 
    \[ 
    \mathbb{E}[V_{\boldsymbol{\theta}}(\mathbf{X}^{\mathrm{All}})] \leq \alpha \opt.
    \] 
\end{proposition}

We prove this through a construction where one patient--provider pair has higher utility than all others, but is unlikely to be selected, leading to near-zero utility (formal proof in Appendix~\ref{app:baseline_proofs}). 
Critically, Offer-All's failure mode can arise even under perfect information ($\eps = 0$). 
Offer-All's weakness is capacity contention, whereas Offer-One's is driven by uncertainty.

Offer-One and Offer-All reveal a fundamental tension in menu design: offering more choice can help patients act on private preference information, but it can also create competition for scarce providers.
We formalize this tension by decomposing welfare loss into two terms that move in opposite directions as choice expands.
For the following decomposition and analysis in Section~\ref{sec:sam}, we measure welfare relative to the exit option.
For a menu $\mathbf{X}$, let $LP_{\theta}(\mathbf{X})$ be the optimal value of the assignment LP restricted to offered edges, and $V_{\theta}(\mathbf{X},
\sigma)$ the realized welfare of the sequential response process.
Then the loss decomposes into two nonnegative parts:
\begin{equation}
\label{eq:decomp}
\opt(\boldsymbol{\theta}) - V_{\boldsymbol{\theta}}(\mathbf{X}, \sigma)
\;=\;
\underbrace{\opt(\boldsymbol{\theta}) - LP_{\boldsymbol{\theta}}(\mathbf{X})}_{\text{menu loss}}
\;+\;
\underbrace{LP_{\boldsymbol{\theta}}(\mathbf{X}) - V_{\boldsymbol{\theta}}(\mathbf{X}, \sigma)}_{\text{coordination
loss}} .
\end{equation}
Menu loss is what the offered graph misses even under perfect coordination, and coordination loss is what sequential selection forfeits within the offered graph. 
There is a tradeoff in offering patients more options.
Offer-One achieves zero coordination loss, as patients' ordering has no impact on performance, but can incur high menu loss as noise increases. 
Conversely, the Offer-All policy achieves zero menu loss, but can suffer from high coordination loss. 
The goal is to offer enough choice to reduce menu loss while limiting coordination loss.  
\section{Scenario-Averaged Marginals (SAM) Algorithm}
\label{sec:sam}
Section~\ref{sec:complexity} findings motivate our central question: how can we selectively expand patient choice beyond Offer-One without creating excessive competition? 
We thus propose the Scenario-Averaged Marginals (SAM) algorithm, which offers multiple options only when utility overcomes provider scarcity.

\subsection{SAM Algorithm}
\label{subsec:stochastic_optimization}
We first quantify provider scarcity through shadow prices. 
Here, shadow prices $\lambda_j$ correspond to capacity constraints in a bipartite matching problem.  

We show that optimal assignments are supported on patient--provider pairs that maximize capacity-adjusted utility $\Delta_{ij}-\lambda_j$. 
Throughout this section, we normalize each patient's exit-option utility to zero and express welfare relative to this baseline. Since this subtracts the same constant \(\sum_i\theta_{i,M+1}\) from every allocation, all welfare differences and regret bounds remain unchanged.
\begin{lemma}[Optimality of Dual Variables]
\label{lem:shadow_price}
Fix any utility matrix $\boldsymbol{\theta} \in [0,1]^{N\times(M+1)}$ and integer capacities $\bar c_j$, and let $\Delta_{ij} := \theta_{ij}-\theta_{i,M+1}$. Consider
\[
\max_{y\ge 0} \sum_{ij} y_{ij}\Delta_{ij} \quad \text{s.t.} \quad \sum_j y_{ij} \leq 1\ \forall i,\quad \sum_i y_{ij}\le \bar c_j\ \forall j\in[M].
\]
Let $\lambda_j$ be optimal dual prices for the capacity constraints.
Then there exists an optimal integral solution \(y^*\), where patient \(i\) is unmatched if \(\max_{\ell\in[M]}(\Delta_{i\ell}-\lambda_\ell)<0\), while patient \(i\) is matched to some \(j\in\arg\max_{\ell\in[M]}(\Delta_{i\ell}-\lambda_\ell)\) if this maximum is positive. 
If the maximum equals zero, either matching to a maximizing provider or remaining unmatched is possible.
\end{lemma}
\begin{proof}
Existence of $\boldsymbol{\lambda}$ follows from LP duality. Let $\mu_i$ be the dual of $\sum_j y_{ij} \leq 1$. By dual feasibility, for every $(i,j)$, $j \in [M]$:
\[
\Delta_{ij} - \mu_i - \lambda_j \le 0, \quad \text{with equality whenever } y^*_{ij}>0,
\]
Rearranging, $\mu_i \ge \Delta_{ij}-\lambda_j$ for every $j$, with equality at the assigned option.
Since the dual objective minimizes \(\sum_i\mu_i+\sum_j\bar c_j\lambda_j\), we may choose
\[
\mu_i = \max(0,\max_{j\in[M]}(\Delta_{ij}-\lambda_j)).
\]
Suppose $y^*_{ij_0}=1$ for some $j_0$. 
By complementary slackness, $\mu_i = \Delta_{ij_0}-\lambda_{j_0}$. For every other $j$, dual feasibility only gives $\Delta_{ij}-\lambda_j \le \mu_i$. 
So $j_0$ must achieve the maximum of $\Delta_{ij}-\lambda_j$ over 
$j\in[M]$. 
Any assigned provider $j_0$ maximizes $\Delta_{ij}-\lambda_j$, with $\Delta_{i,j_0}-\lambda_{j_0} = \mu_i \geq 0$. 
If \(\mu_i>0\), complementary slackness implies \(\sum_j y^*_{ij}=1\), so patient \(i\) must be matched to a maximizing provider.
If \(\max_j(\Delta_{ij}-\lambda_j)<0\), patient \(i\) must remain unmatched. If the maximum equals zero, either assignment to a maximizing provider or remaining unmatched is possible.
\end{proof}

$\Delta_{ij}-\lambda_j$ captures the net value of matching patient $i$ to provider $j$ after accounting for the scarcity of that provider. A capacity-aware menu should thus favor providers with high capacity-adjusted value rather than simply those with high patient utility.

\textit{Scenario Averaging:} 
Computing $\lambda_j$ requires access to the true utilities $\btheta$. 
Different realizations of noise $\delta_{ij}$ result in different values for $\lambda_j$, and so we only have stochastic knowledge of $\lambda_j$. 
To overcome this difficulty, we perturb $\hat{\theta}_{ij}$ to produce $S$ samples $\btheta^{(s)} \sim P_\theta$, so that $\btheta^{(s)} = \bthetahat + \delta^{(s)}$ and $\Delta^{(s)}_{i,j} = \theta^{(s)}_{i,j}-\theta^{(s)}_{i,M+1}$, which represents plausible values for the true preference matrix. 
We can then solve for each $\lambda^{(s)}_j$, then average these to get an overall estimate of the shadow prices, $\lambdahat^{\sam}_j = \frac{1}{S} \sum_{s=1}^{S} \lambda^{(s)}_j$. 
From there, we can similarly compute the net value as
\[
m_{ij} = \frac{1}{S}\sum_{s=1}^{S} \max\{\Delta_{ij}^{(s)}-\hat{\lambda}^{\mathrm{SAM}}_{j},0\}.
\]
Here, $m_{ij}$ is high when provider $j$ remains valuable for patient $i$ after accounting for the provider's estimated scarcity cost across plausible realizations of preferences. 
The primary role of the sampling is to estimate the shadow price of provider $j$ across scenarios, since the menu must be fixed before the true preferences are realized. 

\textit{SAM Algorithm:} 
\begin{algorithm}[t]
\caption{Scenario-Averaged Marginals (SAM) Algorithm}
\label{alg:sam}
\begin{algorithmic}[1]

\STATE \textbf{Input:} Estimated utilities $\bthetahat$, capacities $\bar{c}$, number of scenarios $S$, maximum menu size $k$
\STATE \textbf{Output:} Menu $\mathbf{X} \in \{0,1\}^{N \times M}$

\FOR{$s = 1$ to $S$}
    \STATE Sample perturbed utilities $\boldsymbol{\theta}^{(s)} \sim P_\theta$
    \STATE $\Delta_{ij}^{(s)} \leftarrow \theta^{(s)}_{ij} - \theta^{(s)}_{i,M+1}$ for all $i,j$
    \STATE Let $\lambda_{j}^{(s)}$ be the dual variable for constraint $\sum_{i=1}^{N} y_{ij} \leq \bar{c}_{j}$
\ENDFOR

\STATE $\hat{\lambda}^{\mathrm{SAM}}_j \leftarrow \frac{1}{S} \sum_{s=1}^{S} \lambda_{j}^{(s)}$ for all $j$
\STATE $m_{ij} \leftarrow \left(\frac{1}{S}\sum_{s=1}^{S} \max\{\Delta_{ij}^{(s)}-\hat{\lambda}^{\mathrm{SAM}}_{j},0\}\right)$

\FOR{each patient $i$}
    \STATE Let $A_i$ be the $k$ providers with largest positive $m_{ij}$
    \FOR{each provider $j \in [M]$}
        \STATE $X_{ij} \leftarrow \mathbbm{1}\{j \in A_i\}$
    \ENDFOR
\ENDFOR

\STATE \textbf{return} $\mathbf{X}$

\end{algorithmic}
\end{algorithm}
We present SAM in Algorithm~\ref{alg:sam}. 
After estimating $\lambdahat^{\sam}_j$ and computing $m_{ij}$, SAM sets $X_{ij}=1$ for the top $k$ largest positive $m_{ij}$. 
SAM only offers choice when the net value is positive, thereby ensuring that scarce providers are reserved for the patients who benefit the most. 
Unlike Offer-One, SAM can offer multiple alternatives for a single patient, and unlike Offer-All, SAM can withhold providers whose net value does not justify their opportunity cost, limiting competition for scarce capacity.

\subsection{General Loss Bounds}
\label{subsec:general_loss}
We next formalize bounds on the menu and coordination loss that hold for any menu $\mathbf{X}$. We then establish sufficient market conditions for low loss in SAM-constructed menus. 
Throughout, let $(\mu^*, \lambda^*)$ denote an optimal dual pair of the LP in Lemma~\ref{lem:shadow_price} under the true utilities $\boldsymbol{\theta}$, so that $\opt = \sum_i \mu_i^* + \sum_j \lambda_j^* \bar{c}_j$. 
When this pair is not unique, any optimal pair may be used in the lemmas below, and the conditions in Section~\ref{subsec:sam_friendly_markets} need only hold for one such pair.

\begin{lemma}[Menu Loss Bound]
\label{lem:menu_loss}
For every menu $\mathbf{X}$ and every $t \geq 0$, define 
\[
E_t(\mathbf{X}) := \{(ij) : X_{ij}=1,\ \mu_i^* + \lambda_j^* - \Delta_{ij} \leq t\}
\]
and 
\[
\Def_t(\mathbf{X})
:=
OPT-
\max_{\substack{y\text{ feasible}\\
\operatorname{supp}(y)\subseteq E_t(\mathbf{X})}}
\sum_{ij}y_{ij}(\mu_i^*+\lambda_j^*).
\]
Then we can bound the menu loss as 
\[
OPT - LP_{\theta}(\mathbf{X}) \;\le\; t\,C + \mathrm{Def}_t(\mathbf{X}).
\]
\end{lemma}
Lemma 2 bounds the menu loss in terms of the availability of near-optimal edges. 
We prove this by decomposing the loss into the value lost from using low-slack edges rather than zero-slack optimal edges, and the value lost because the menu does not contain enough such low-slack edges.
The former simplifies to a $tC$ term, while the latter term simplifies to $\Def_{t}(\mathbf{X})$ (Proof in Appendix~\ref{app:sam_proofs}). 
Next, we prove a bound for the coordination loss. 

\begin{lemma}[Coordination Loss Bound]
\label{lem:coordination_loss}
Fix the menu $\mathbf{X}$ and a response order $\sigma$.
Let $(\mu^\mathbf{X}, \lambda^\mathbf{X})$ be an optimal dual pair of the menu-restricted LP, let $F$ be the set of patients who are matched, and let $f_j$ be the realized fill of provider $j$,
and for each matched patient let $g^\mathbf{X}_{i,s_i} := \mu^\mathbf{X}_i +
\lambda^\mathbf{X}_{s_i} - \Delta_{i,s_i} \ge 0$ be the realized edge's slack
under the menu duals. Then, exactly,
\[
LP_{\theta}(\mathbf{X}) - V(\mathbf{X},\sigma)
\;=\;
\sum_{i \notin F} \mu^\mathbf{X}_i
\;+\;
\sum_{j} \lambda^\mathbf{X}_j (\cbar_j - f_j)
\;+\;
\sum_{i \in F} g_{i,s_i}^{\mathbf{X}},
\]
and each of the three terms is nonnegative.
\end{lemma}

Lemma~\ref{lem:coordination_loss} exactly characterizes coordination loss by how sequential choice can differ from the optimal allocation. There are three sources of sequential inefficiency captured in the three terms above: (i) patients with positive value under the menu LP may remain unmatched, (ii) valuable provider capacity may remain unfilled, and (iii) matched patients may select edges that are suboptimal relative to the menu LP. 
The third term can arise because patients choose according to their own utilities without internalizing provider opportunity costs (proof in Appendix~\ref{app:sam_proofs}).
We can combine these two lemmas into a general corollary that bounds the overall loss for a menu $\mathbf{X}$. 

\begin{corollary}[Overall Loss Bound]
\label{cor:sufficient_conditions}
Fix a menu $\mathbf{X}$ and $t \geq 0$. Then
\[
\begin{aligned}
OPT - V(\mathbf{X},\sigma)
\leq\;&
tC + \Def_t(\mathbf{X}) \\
&+ \sum_{i\notin F}\mu_i^\mathbf{X}
+ \sum_j \lambda_j^\mathbf{X}(\bar c_j-f_j)
+ \sum_{i\in F}g^\mathbf{X}_{i,s_i}.
\end{aligned}
\]
\end{corollary}

This follows directly from Lemma~\ref{lem:menu_loss} and Lemma~\ref{lem:coordination_loss}. 
Corollary~\ref{cor:sufficient_conditions} demonstrates that two requirements are needed for good menu design. 
First, the menu must contain enough near-optimal edges to support a high-value allocation. Second, sequential patient choice over those edges must not substantially disrupt that allocation.

\subsection{Structural Conditions Favoring SAM}
\label{subsec:sam_friendly_markets}
We next use Corollary~\ref{cor:sufficient_conditions} to identify structural conditions under which SAM achieves provably small regret. 
Rather than characterize all settings in which SAM performs well, our goal is to isolate conditions that control the two sources of loss identified above. 
Menu loss is small when SAM's menus contain sufficiently valuable alternatives (\ref{cond:dual}--\ref{cond:balanced}). Coordination loss is small when those alternatives do not create competition for scarce capacity (\ref{cond:slack}).

\textit{Menu loss:} We first state two conditions under which SAM achieves low menu loss.
\begin{enumerate}[label=(C\arabic*), ref=C\arabic*, leftmargin=*]
\item \textbf{Dual Stability} - Dual stability holds with constant $L_\lambda \ge 0$ if $\|\lambda^{\sam} - \lambda^*\|_\infty \leq L_\lambda \epsilon$.  \label{cond:dual}
\item \textbf{Balanced Demand} - Let $I^*:=\{i:\mu_i^*>0\}$ denote the patients with positive value in the unrestricted LP. For each $i\in I^*$, there exists at least one $m_{ij} > 0$, and let $j^{\sam}(i)$ denote the provider with the largest positive $m_{i,j}$, with $D_j:=\left|\{i\in I^*:j^{\sam}(i)=j\}\right|$. Then $D_j\leq \bar c_j$, with equality whenever $\lambda_j^*>0$. \label{cond:balanced}
\end{enumerate}

Intuitively, Condition \ref{cond:dual} captures the idea that we can accurately estimate $\lambda^*$ from finite samples.
This condition is plausible when small utility perturbations do not result in large changes in the optimal dual solution. 
Condition \ref{cond:balanced} captures the idea that demand is not concentrated: no provider is over-demanded by patients' top-ranked options, so those options alone form a feasible allocation. 

Under dual stability, we can first show that dual slack is small. 
\begin{proposition}[Dual Stability Guarantees Low Slack]
\label{prop:slack}
Let
\[
r_{ij} := \Delta_{ij}-\lambda_j^*,
\qquad
r_{ij}^+ := \max\{r_{ij},0\},
\qquad
g_{ij} := \mu_i^*-r_{ij}.
\]
Under dual stability, with $\eta := 4 + L_\lambda$: (i)
$|m_{ij} - r^+_{ij}| \le \eta\eps$ for all $i,j$; (ii) every designated edge $(ij^{\sam}(i))$ has dual slack
\[
g_{ij^{\sam}(i)}\leq 3\eta\eps.
\]
\end{proposition}
We prove this by showing that the gap between $\Delta^{(s)}$ and $\Delta$ are within $O(\epsilon)$ of each other, and the same holds for $\lambda^{\mathrm{SAM}}$ and $\lambda^*$, then relating these errors to the dual slack (full proof in Appendix~\ref{app:sam_proofs}).  
As a result, each patient's designated provider is approximately optimal after accounting for the opportunity cost of provider capacity.
The remaining challenge is the need for sufficiently many such edges to guarantee good allocation. 
If all selected pairs are concentrated on a few providers, then such an allocation might not exist.
The balanced demand condition ensures that high-quality edges are not overly concentrated on any provider. 
We then invoke Lemma~\ref{lem:menu_loss} with $t=3 \eta \epsilon$ and bound the $\mathrm{Def}_t(\mathbf{X})$ to get the bound. 

\begin{proposition}[Menu Loss of SAM]
\label{prop:sam_menu_loss}
Suppose that dual stability and balanced demand hold true. Then
\[
OPT-LP_{\theta}(\mathbf{X}^{\sam})
\le
3\eta\epsilon C
=
O(\epsilon C),
\]
where $\eta=4+L_\lambda$.
\end{proposition}

We prove this by assigning each $i\in I^*$ to its designated provider, which is feasible by \ref{cond:balanced}, with loss bounded by $3\eta\eps$ per edge by Proposition~\ref{prop:slack} (full proof in Appendix~\ref{app:sam_proofs}). Therefore, when SAM's shadow prices are stable and its high-value alternatives are sufficiently dispersed across provider capacity, restricting patients to SAM's menus incurs only $O(\epsilon C)$ menu loss. 

\textit{Coordination loss:}
We next provide a condition under which SAM has low coordination loss:
\begin{enumerate}[label=(C\arabic*), ref=C\arabic*, leftmargin=*, resume]
    \item \textbf{Slack-Supported Alternatives} -  Let \(X^{\mathrm{SAM}}\) be the menu returned by Algorithm~\ref{alg:sam}. For each provider $j$, let $T_j := \left|\{i : X^{\sam}_{ij}=1\}\right|$ denote the total number of patients to whom provider $j$ is offered. Then let $T_j \leq \bar c_j$ for all $j$. \label{cond:slack}
\end{enumerate}
Under Condition~\ref{cond:slack}, additional choice cannot create capacity-driven displacement.
Such a condition is sufficient rather than necessary, and is meant to illustrate a sufficient regime where coordination loss has minimal impact. 
We can then invoke Lemma~\ref{lem:coordination_loss}; here the three terms are all zero because there is no competition for any provider. 
Combined with the menu loss bound, this yields:
\begin{proposition}[Overall Loss of SAM]
\label{prop:sam_coordination}
Under dual stability, balanced demand, and slack-supported alternatives we have that 
\[
OPT - V(\mathbf{X}^{\sam},\sigma)
\leq
3\eta\epsilon C
=
O(\epsilon C)
\]
for every response order $\sigma$.
\end{proposition}

To show this, recall that coordination loss occurs because patients respond sequentially; an early patient may select a provider that would have been more valuable to a later patient.
The key distinction is therefore not simply how many options SAM offers, but whether its additional options create competition for already scarce capacity.
A simple sufficient condition eliminates this coordination failure entirely: each provider has enough capacity to accommodate every patient to whom it is offered.
Thus, regardless of the response order or which offered providers patients select, no provider can become unavailable to another patient who was also offered that provider (full proof in Appendix~\ref{app:sam_proofs}).

Propositions~\ref{prop:sam_menu_loss} and~\ref{prop:sam_coordination} identify a simple SAM-friendly market structure in which choice carries no coordination cost and menus are near-optimal. 
Condition~\ref{cond:slack} implies that $\sum_i L_i \le C$. 
Combined with Condition~\ref{cond:balanced}, a scarce provider appears on a patient's menu only when it is that patient's first option under SAM's scores, so any additional choice is among providers whose shadow price is zero.
Condition~\ref{cond:slack} is intentionally strong, as it allows us to isolate a setting where the informational benefit of additional choice can be obtained without coordination loss. 
Outside this regime, SAM may still improve welfare even when its menus create competition for scarce capacity, and we study these more congested settings computationally in Section~\ref{sec:empirical}. 

\subsection{Comparing SAM and Offer-One}
\label{subsec:comparing_sam_offer_one}
Next, we present two examples to illustrate when additional choice provides value. 
In the first example, offering choice improves performance because patients have multiple good provider options, so adding choice gains information on patient preferences. 
In the second example, adding choice decreases performance due to competition for a provider, leading to high coordination loss. 
\begin{example}[The Value of Tailored Choice]
\label{ex:sam-choice}
Consider a setting with $N=2$ patients and $M=3$ unit-capacity providers, with $k=2$.
Suppose that the platform observes
\[
\hat{\boldsymbol{\theta}}
=
\begin{pmatrix}
1/2 & 1/2 & 1/2\\
3/4 & 1/4 & 1/4
\end{pmatrix},
\]
and that the exit option has utility $0$.
Suppose that $\delta_{ij}\overset{\mathrm{i.i.d.}}{\sim}\mathrm{Unif}[-\epsilon,\epsilon]$. 
Then Offer-One assigns provider $2$ (or $3$) to patient $1$ and provider $1$ to patient $2$, yielding $\mathbb{E}[V^{\mathrm{Offer-One}}]=\frac54$.
For concreteness, suppose that SAM selects the minimum norm shadow prices when multiple duals are possible. 
We note that the expected duals are $\mathbb{E}[\lambda_1] = \mathbb{E}\big[\left(\theta_{11}-\max\{\theta_{12},\theta_{13}\}\right)_+\big] = \frac{\epsilon}{6}$, and $\mathbb{E}[\lambda_2] = \mathbb{E}[\lambda_3] = 0$. 
As a result, SAM offers providers $2$ and $3$ to patient $1$, while provider $1$ remains available to patient $2$.
We then have that 
\[
\mathbb{E}[V^{\sam}]
=
\frac54+\frac{\epsilon}{3}
>
\frac54
=
\mathbb{E}[V^{\mathrm{Offer-One}}].
\]

Offer-All provides the same informational benefit but additionally exposes provider $1$ to patient $1$.
Averaging over preferences and response order gives
\[
\mathbb{E}[V^{\mathrm{Offer-All}}]
=
\frac76+\frac{17}{36}\epsilon
<
\frac54+\frac{\epsilon}{3}
=
\mathbb{E}[V^{\sam}],
\]
where the inequality holds for $0<\epsilon<1/8$.
\end{example}

The example illustrates the value of tailored choice.
Offer-One is overly conservative because it prevents patient $1$ from using their private preference information to choose between two similar alternatives.
Offer-All provides this useful choice but also exposes a contested provider, creating a coordination cost.
SAM separates these two effects by offering choice among useful substitutes while withholding capacity whose opportunity cost is high.
However, richer menus need not always improve performance.
The following modification illustrates how additional choice can instead create costly competition for scarce capacity.

\begin{example}[When Additional Choice Hurts]
\label{ex:sam-hurts}
Consider a setting with $N=M=2$ unit-capacity providers, with $k=2$, and
\[
\hat{\boldsymbol{\theta}}
=
\begin{pmatrix}
1/2 & 1/2\\
3/4 & 1/4
\end{pmatrix},
\]
with the same i.i.d.\ noise as before. 
Offer-One assigns provider $2$ to patient $1$ and provider $1$ to patient $2$, yielding $\mathbb{E}[V^{\mathrm{Offer-One}}]=\frac54$.

Again, suppose that SAM selects the minimum norm dual. 
For each realization, the minimum optimal dual has $\mathbb{E}[\lambda_1]=\mathbb{E}\big[(\theta_{11}-\theta_{12})_+\big] = \frac{\epsilon}{3}$ in expectation. 
As a result, with $k=2$, SAM offers both providers to patient $1$.
The same holds for patient $2$, so SAM coincides with Offer-All.
If patient $1$ responds first, they may consume provider $1$, despite provider $1$ being substantially more valuable to patient $2$.
Averaging over preferences and response order gives
\[
\mathbb{E}[V^{\sam}]
=
\frac98+\frac{\epsilon}{6}
<
\frac54
=
\mathbb{E}[V^{\mathrm{Offer-One}}],
\]
for $0<\epsilon<1/8$.
\end{example}

The two examples demonstrate that the value of choice is instance-dependent. 
When patients have multiple good options, choice can be beneficial, but under competition for providers, choice can increase coordination loss. 
We next explore when each of these situations occur in practice. 
\section{Computational Experiments}
\label{sec:empirical}
We characterize the regimes where larger menus are useful using a semi-synthetic environment designed to capture patient–provider matching dynamics in large U.S. healthcare platforms. 
Our goal is not to estimate the effect of deploying SAM in a particular health system, but to use a healthcare-calibrated semi-synthetic environment to test the mechanisms identified by the theory and characterize the market conditions under which patient choice is valuable.
We first describe our semi-synthetic environment, then characterize the impact of different factors, including noise level, congestion, and provider capacity.

\subsection{Experimental Setup}
\label{subsec:experimental_setup}
\textit{Data and Semi-synthetic Environment:}
We construct a semi-synthetic simulation to better understand how operational factors impact performance. 
We select a default set of hyperparameters to mimic the patient and provider counts in a real-world health system, and let $N=1225$ and $M=700$ based on numbers from a partner healthcare system.
We select default values for other parameters as follows: $\mathcal{D}$ is uniform over permutations of $[N]$, $k=25$, $S=10$, $\epsilon=0.25$, $\bar{c}_{j} \sim \mathrm{Poisson}(\lambda)$, where we let $\lambda=1$ by default. 
We exclude any providers with zero capacity $\bar{c}_j=0$. 

Next, we construct $\boldsymbol{\hat{\theta}}$ as a proxy for real patient preferences by incorporating geographic proximity and comorbidity alignment. 
Prior work demonstrates that geographic proximity is the top factor impacting patient utilities~\citep{washington_distance,distance_rural}. 
Comorbidities influence patient preferences because patients prefer providers with expertise tailored to their specific conditions or care needs.
While these two features do not perfectly capture patient preferences, they provide a reasonable estimate that can be substituted with richer data-driven or survey-based predictive models. 
We describe the exact parameterization of $\boldsymbol{\hat{\theta}}$ in Appendix~\ref{app:experimental_details}.

\textit{Policies:}
In addition to Offer-One, we evaluate two other baselines: Top-$k$, which offers each patient the $k$ providers with the largest estimated utility, and Random, which lets $\mathbf{X}_i$ be a random uniform $k$-subset of providers. 
Top-$k$ serves as the menu-constrained analog of Offer-All: it gives each patient the maximum allowable number of options without accounting for provider capacity, and coincides with Offer-All when $k=M$.
We additionally compare with two other baselines in Appendix~\ref{app:stronger_baselines} and find similar results. 

\textit{Evaluation Protocol:}
In each experiment, we construct menus and then simulate sequential patient selection of providers according to their true $\boldsymbol{\theta}$. 
For each realization, we sample a new random permutation $\boldsymbol{\sigma}$, and compute metrics for patient utilities. 
All experiments are averaged across nine seeds, and for each seed, we evaluate each policy across 25 independently sampled preference realizations $\boldsymbol{\theta} \sim \mathcal{P}_{\boldsymbol{\theta}}$ and patient response orders $\boldsymbol{\sigma} \sim \mathcal{D}$.
For all figures, we report the standard error across the nine seed-level averages. 
We compare policies according to the normalized utility, which is the realized total utility $\sum_{i} \theta_{i,s_{i}}$ divided by $\opt$. 
Here, $\opt$ observes the realized patient preferences $\btheta$ but not the response order $\sigma$. 

\subsection{When Is Choice Valuable?}
\label{subsec:when_is_choice_valuable}
Our theoretical results suggest that the value of choice depends on two competing forces: its ability to hedge against uncertainty in patient preferences and the coordination loss created when patients compete for scarce provider capacity. 
We therefore vary three features of the market: preference uncertainty, which determines how much additional information patients hold; provider capacity, which determines the cost of decentralized choice; and preference differentiation, which determines how consequential it is to select the right provider. 
We then examine how SAM adapts its menus across these regimes.

\begin{figure*}
\centering
\includegraphics[width=\textwidth]{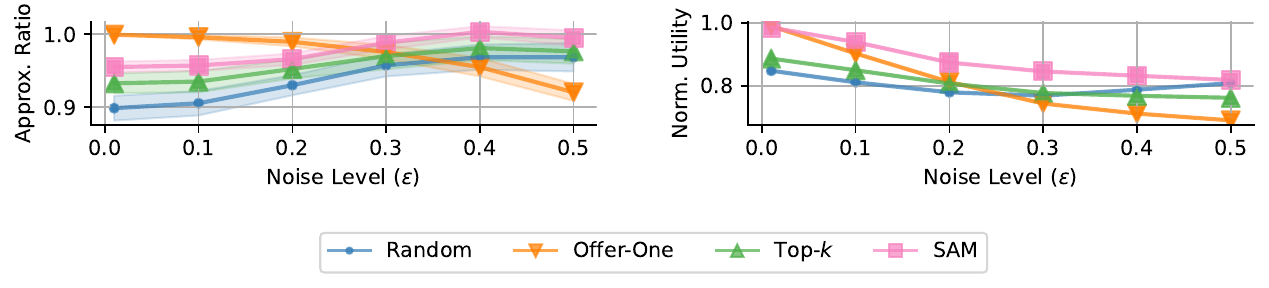}
\caption{Impact of preference uncertainty on matching performance. As preference estimates become noisier, SAM increasingly outperforms Offer-One, demonstrating the growing value of patient choice under preference uncertainty. Left: approximation ratio of policies as $\epsilon$ increases with $N=10$ and $M=5$, where approximation ratio is compared to a sample--average optimal MILP policy. Right: normalized utility of policies as $\epsilon$ increases with $N=1225$ and $M=700$.}
\label{fig:noise_sweep}
\end{figure*}

\textit{Noise Level:} In Figure~\ref{fig:noise_sweep}, we show that as preference uncertainty increases, the value of patient choice grows. 
On the left, we consider a smaller problem instance with $N=10$, $M=5$, $k=2$ and compute the approximation ratio relative to a sample-average optimal MILP policy (details in Appendix~\ref{app:milp}). 
At $\epsilon=0.01$, SAM has a 4\% lower approximation ratio than Offer-One, but by $\epsilon=0.5$, SAM has an 8\% better approximation ratio than Offer-One. 
Similarly, on the right we compare the normalized utility for Offer-One and SAM at $N=1225$ and $M=700$. 
At $\epsilon=0.01$, SAM and Offer-One are within 0.5\% of each other, but by $\epsilon=0.5$, SAM is 19\% better than Offer-One. 
This occurs because accurate estimates make it effective to commit to a single provider, while under greater uncertainty, SAM can offer several plausible alternatives and allow patients to act on unseen preference information.
While Random and Top-$k$ perform worse than Offer-One at $\epsilon=0.01$, they outperform Offer-One for $\epsilon \geq 0.3$. 

\begin{figure*}
\centering
\includegraphics[width=\textwidth]{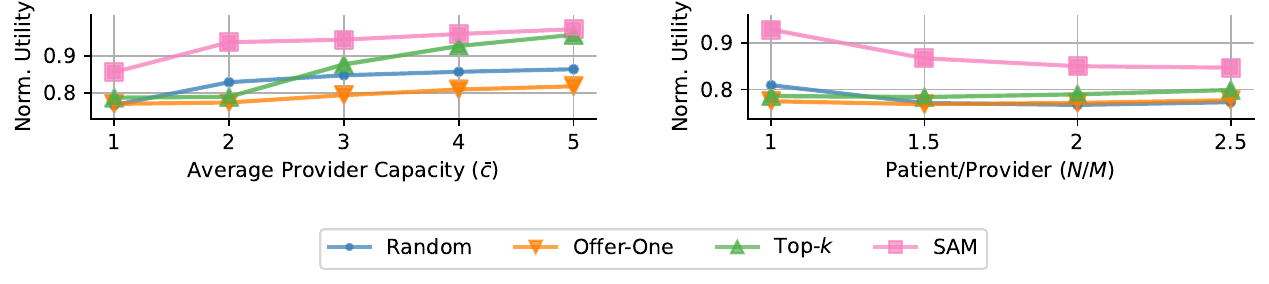}
\caption{Impact of market congestion on matching performance. Left: increasing average provider capacity reduces competition and increases the relative value of choice. Right: SAM performs best at lower patient--provider ratios, while the other policies are comparatively insensitive to this ratio.}
\label{fig:market_dynamics}
\end{figure*}

\textit{Market Congestion:} To explore congestion effects, we plot performance when varying the average provider capacity $\bar{c}$ (by varying $\lambda$) and $N$ (while fixing $M=700$) in Figure~\ref{fig:market_dynamics}. 
As average provider capacity increases, SAM and Top-$k$ both improve; when $\mathbb{E}[\bar{c}_j] = 1$ (corresponding to $\lambda=1$), SAM is only 11\% better than Offer-One and Top-$k$ is 2\% better than Offer-One. 
However, when $\mathbb{E}[\bar{c}_j]=5$ (corresponding to $\lambda=5$), both SAM and Top-$k$ are at least 17\% better than Offer-One.  
Increasing provider capacity $\bar{c}$ reduces inter-patient competition, which incentivizes larger menus due to decreased coordination loss. 
Meanwhile, SAM performs best for lower patient--provider ratios, though the effect on other policies is small, as non-SAM policies only vary by 4.5\% between $N/M=1$ and $N/M=2.5$. 
Increasing or decreasing $N/M$ changes the number of patients, but still maintains inter-patient competition, and so has little impact on the relative performance of different policies.

\begin{figure*}
\centering
\includegraphics[width=\textwidth]{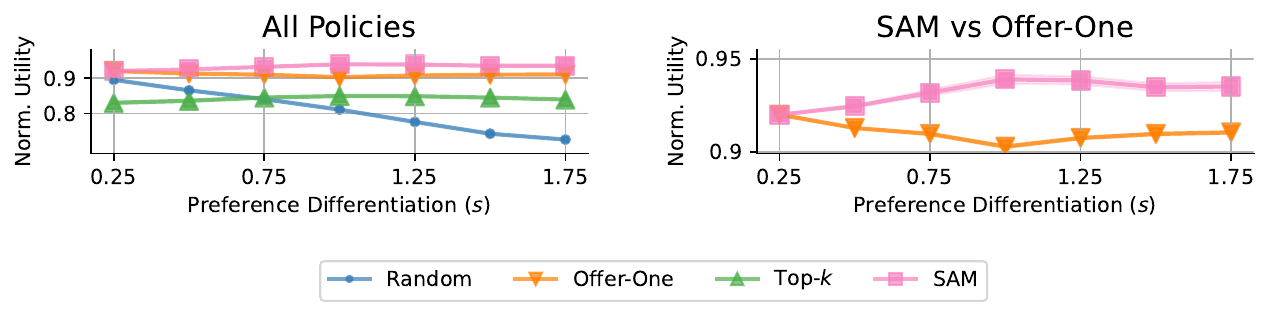}
\caption{Impact of preference differentiation. We scale each patient's estimated utility gaps from their highest-rated provider by $s$, with larger $s$ corresponding to more differentiated utilities. SAM maintains consistently high normalized utility across differentiation levels, while Offer-One falls then rises as preferences become more differentiated.}
\label{fig:utility_vs_spread_top}
\end{figure*}

\textit{Preference Differentiation:} 
We next study how the spread of utility values impacts performance. 
Starting from the estimated preference matrix $\bthetahat$ and $\epsilon=0.1$, we scale each patient's utility gaps relative to their highest-valued provider by a factor $s$, such that $s<1$ compresses utilities and creates more similarly valued providers, while $s>1$ increases differentiation. 
We then sample the true preference matrix $\btheta$ according to our standard noise model, holding the noise level fixed. Figure~\ref{fig:utility_vs_spread_top} shows that the value of choice depends non-monotonically on preference differentiation. 
When preferences have minimal differentiation (small $s$), the difference in estimated utilities is small, leading Offer-One and SAM to achieve similar performance. 
As differentiation increases, SAM performs better, as it allows patients to select based on their realized preferences by offering multiple options to patients. 
When patients are highly differentiated (large $s$), noise has little impact on the performance, leading to a shrinking of the gap between SAM and Offer-One. 
As a result, the value of choice is highest when preferences are differentiated enough so all providers do not look the same, yet also when noise can change the ordering of providers between $\hat{\btheta}$ and $\btheta$.
Top-$k$ is flat across values of preference differentiation because its menus do not change, leading to similar coordination loss. 
These results complement Section~\ref{subsec:sam_friendly_markets}: the gains from richer menus are greatest when preference differences are large enough for assignment errors to matter, but not so large that they outweigh preference uncertainty.

\subsection{How does SAM Adapt?}
\label{subsec:sam_adapt}
\begin{figure*}
\centering
\includegraphics[width=\textwidth]{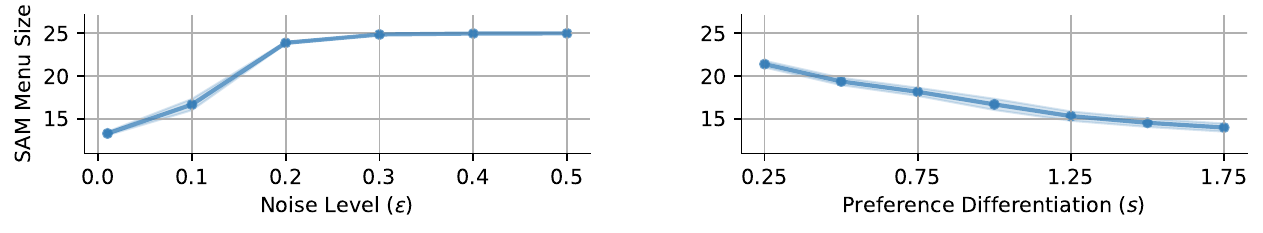}
\caption{Menu size for SAM as a function of noise (left) and preference differentiation (right). Increasing noise leads to larger menus, while increasing preference differentiation leads to smaller menus.}
\label{fig:sam_menu_size}
\end{figure*}

\textit{Menu Size:} 
In Figure~\ref{fig:sam_menu_size} (left), we show that SAM dynamically varies menu size depending on noise level $\epsilon$. 
When $\epsilon = 0.01$, SAM has a menu size of 13, but by $\epsilon=0.2$, this grows to 24. 
SAM offers a provider when its capacity-adjusted value is positive in at least one scenario. Thus, providers with zero shadow price and positive estimated surplus always appear on menus, even at low noise. Scarce providers are offered when some scenario ranks them above their shadow price, which becomes more likely as noise increases and thus favors larger menus.
Figure~\ref{fig:sam_menu_size} (right) additionally shows that SAM's menu size decreases with increasing preference differentiation, from 21 at $s=0.25$ to 14 at $s=1.75$. 
Larger menus become unnecessary when preferences are sufficiently differentiated. Therefore, two primary levers control SAM's menu sizes: noise $\epsilon$ expands them, and preference differentiation $s$ contracts them.

\begin{figure*}
\centering
\includegraphics[width=\textwidth]{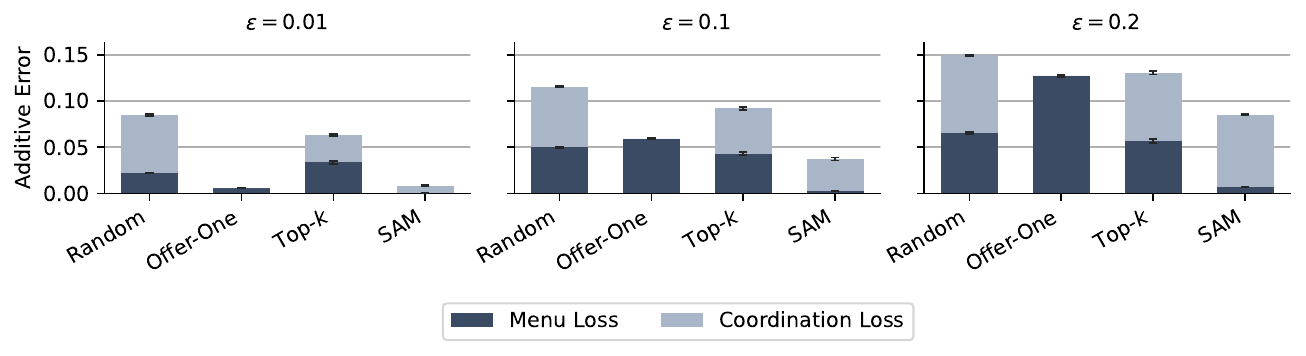}
\caption{Decomposition of the additive regret per patient into menu loss and coordination loss at $\epsilon \in \{0.01,0.1,0.2\}$. Increasing $\epsilon$ leads to larger menu loss for Offer-One,  whereas policies with larger menus have lower menu loss but a higher coordination loss. SAM reduces menu loss through tailored choice, while incurring coordination loss comparable to other larger-menu policies.}
\label{fig:error_decomposition}
\end{figure*}

\textit{Loss Decomposition:} To better understand the interplay between choice and noise, Figure~\ref{fig:error_decomposition} decomposes the additive regret per patient into the menu loss and coordination loss terms studied in Section~\ref{sec:sam}. 
The decreasing performance of Offer-One with increasing noise in Figure~\ref{fig:noise_sweep} is driven by menu loss in cases where a patient's best option falls outside of the menu.
While menu loss increases for Random and Top-$k$, it increases at a slower rate compared to Offer-One; for Random it increases by a factor of 2.9 from $\epsilon=0.01$ to $\epsilon=0.2$, and by a factor of 1.7 for Top-$k$, while for Offer-One, it increases by a factor of 22. 
SAM's primary benefit stems from low menu loss due to tailoring menus so each patient has valuable offers. 
We see that all three menu-based approaches incur similar coordination losses at $\epsilon=0.2$,  with SAM, Top-$k$, and Random varying by only 9\%. 
Structured choice proves increasingly important at larger $\epsilon$ to combat the rapidly increasing menu loss that arises under noisy preferences. 
Going back to Section~\ref{subsec:sam_friendly_markets}, while condition C3 does not hold, leading to a larger coordination loss, the menu loss remains small, partially because patients have differentiated preferences. 
In Appendix~\ref{app:dual-assumptions}, we further decompose the coordination loss.

\subsection{Impact of Choice Beyond Aggregate Welfare}
\label{subsec:beyond_utility}

\begin{table*}[]
    \centering 
\begin{tabular}{lcc}
  \toprule
  Method & \# Available Choices at Response & \#Choices in $k$ Most Preferred\\
  \midrule
  Random     & $8.42 \pm 0.08$ & $0.22 \pm 0.00$ \\
  Offer-One  & $1.57 \pm 0.01$ & $0.19 \pm 0.00$ \\
  Top-$k$    & $5.85 \pm 0.06$ & $1.29 \pm 0.02$ \\
  SAM        & $8.53 \pm 0.08$ & $1.24 \pm 0.03$ \\
  \bottomrule
  \end{tabular}
    
    \caption{Number of choices per patient and the average number of available choices in a patient's $k$ most preferred providers (based on true preferences). SAM and Random offer the most choices, but Random tends to offer low-quality choices, while SAM and Top-$k$ offer high-quality choices.}

\label{tab:choice_quality}
\end{table*}
Beyond utility, we examine how offering larger menus impacts the quality of choices perceived by patients and the geographic equity of the resulting matches. 

\textit{Choice Quality:} 
In Table~\ref{tab:choice_quality}, we compare the number of choices available when a patient comes to match, along with the average number of such choices in the patient's true top-$k$ across the original menu. 
The former captures the quantity of choices while the latter captures the quality of choices, two key components of meaningful choice. 
SAM and Random offer the most choices, with Random offering 8.4 choices on average and SAM offering 8.5 choices (including the exit option). 
However, the quality of the choices differs between the two; only $0.22$ choices from Random are in a patient's $k$ most preferred providers (on the true $\boldsymbol{\theta}$), while approximately $1.2$ choices are present for both Top-$k$ and SAM.  

\textit{Geographic Equity:} We next examine geographic equity. 
In patient--provider matching, ensuring geographic equity is important so rural patients get access to high-quality providers.
We focus on ZIP code population density as a proxy for rurality, since lower-density areas typically have fewer nearby providers and face greater barriers to matching. 
Because geographic proximity is one component of our constructed patient utilities, we interpret this analysis as characterizing the geographic distribution of welfare under our semi-synthetic preference model, rather than establishing a general fairness property of SAM.
Similar analyses can be repeated for other fairness-related quantities.

\begin{figure*}
\centering
 \includegraphics[width=\textwidth]{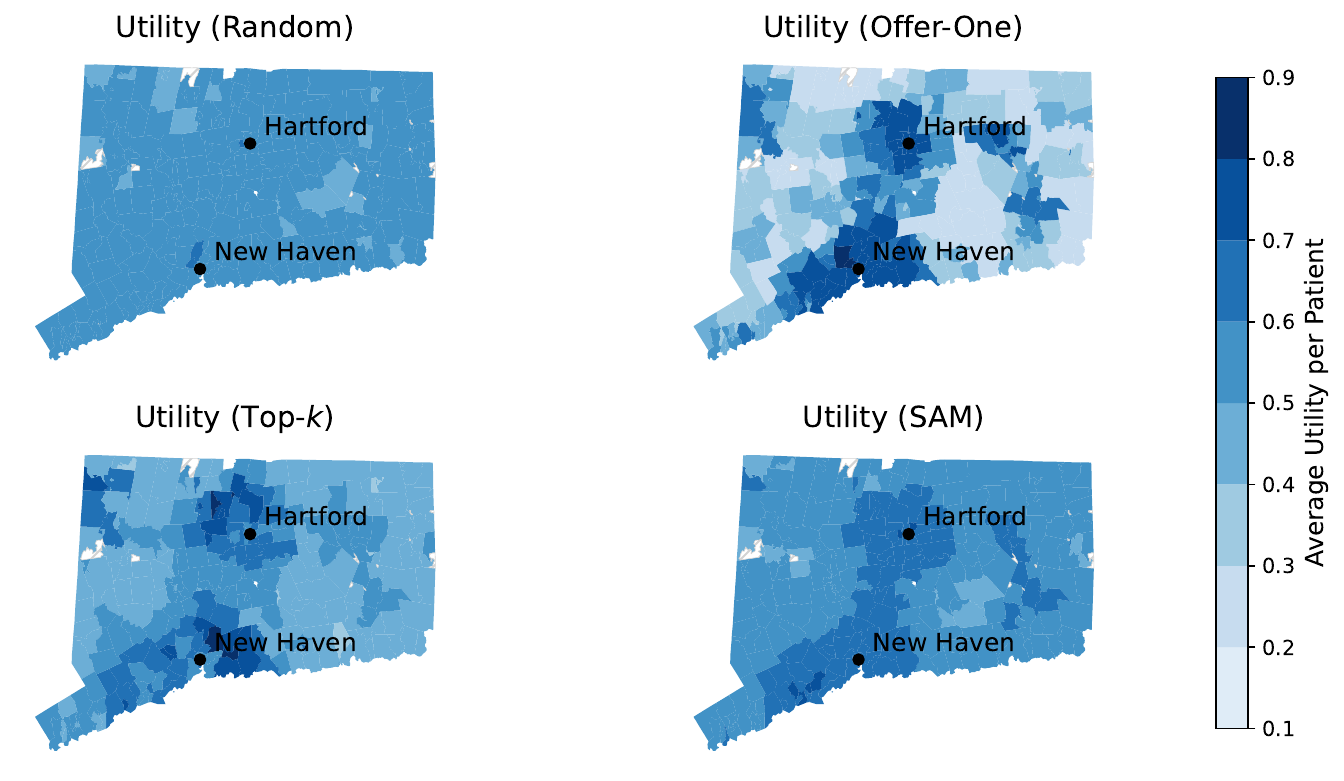}
\caption{Mean patient utility by ZIP code under each policy. Offer-One and Top-$k$ exhibit substantially greater geographic concentration in utility, while SAM and Random produce stronger outcomes across a broader set of ZIP codes.}
\label{fig:utility_by_zipcode}
\end{figure*}

We visualize geographic patterns in Figure~\ref{fig:utility_by_zipcode}, plotting the average utility per ZIP code for the four policies. 
SAM improves lower-tail ZIP-code utility and produces a more geographically uniform distribution of utility than Offer-One and Top-\(k\), whose utility is concentrated around Hartford and New Haven, leaving patients elsewhere near exit utility. 
SAM improves on Offer-One in 66\% of ZIP codes, with the largest gains in rural areas: utility increases by 53\% in the least densely populated quartile of ZIP codes, compared with only 11\% in the most populated quartile.
This is not an explicit design objective of SAM, but rather a byproduct of offering tailored choice. 
Such a pattern is consistent with SAM's menus steering patients away from high-demand urban providers that would otherwise be rapidly depleted and toward nearby alternatives that are available and still high-utility.

Random also spreads utility evenly, but at a lower level: SAM improves on Random in 92\% of ZIP codes.
SAM's advantage here comes from better matching in two large cities, Hartford and New Haven, where SAM's utility is 14\% higher than Random, versus within 4\% for ZIP codes more than 40 miles from either. 
More broadly, these results demonstrate that tailored choice can preserve access to high-quality alternatives across geographic regions while limiting competition for scarce providers in high-demand areas.

\begin{table}[t]
\centering
\begin{tabular}{lcc}
\toprule
Method & 25th Pct. Utility & Gini Coefficient\\
\midrule
Random     & $0.512 \pm 0.003$ & $0.044 \pm 0.001$ \\
Offer-One  & $0.257 \pm 0.001$ & $0.259 \pm 0.002$ \\
Top-$k$    & $0.466 \pm 0.006$ & $0.120 \pm 0.003$ \\
SAM        & $0.540 \pm 0.003$ & $0.068 \pm 0.001$ \\
\bottomrule
\end{tabular}
\caption{Fairness metrics across algorithms for ZIP-code level utility: 25th
percentile of the per-ZIP mean utility (higher is better) and Gini coefficient (lower
is better). SAM maximizes the 25th percentile outcome for ZIP codes, while Random minimizes the spread of utilities across ZIP codes. }
\label{tab:fairness_metrics_utility}
\end{table}

Quantitatively, we evaluate fairness through two metrics: the 25th-percentile outcome across ZIP codes and the Gini coefficient (Table~\ref{tab:fairness_metrics_utility}). 
The 25th-percentile utility captures outcomes in lower-performing geographic areas, which in our setting disproportionately correspond to rural ZIP codes. 
The Gini coefficient captures disparities in utility across ZIP codes, with lower values indicating a more equitable geographic distribution of match quality.
We observe that Offer-One significantly underperforms other policies: the 25th-percentile utility of 0.257 essentially matches the exit option of 0.25, suggesting that at least a quarter of ZIP codes receive almost no value from matching.
We find that SAM achieves the highest 25th-percentile utility and is 5\% better than Random, next best policy. Random achieves the best Gini coefficient, followed by SAM, obtaining the most geographically uniform outcomes. 
However, this greater equality does not correspond to better outcomes, as Random's lower Gini is accompanied by lower utility than SAM in both the lower tail and overall. Therefore, the geographic benefits of patient choice depend not simply on distributing matches more evenly, but on providing high-quality choice while managing scarce capacity.

Overall, these results show that personalized menus have value beyond aggregate utility. 
By accounting for provider competition, SAM offers patients a larger set of high-quality providers while also improving outcomes for geographically underserved populations, even without fairness as an explicit objective. 
Tailored choice redistributes access away from oversubscribed providers while preserving high-quality alternatives.

\subsection{Robustness to Modeling and Design Choices}
\label{subsec:other_choice}

\begin{figure*}
    \centering 
\includegraphics[width=0.5\textwidth]{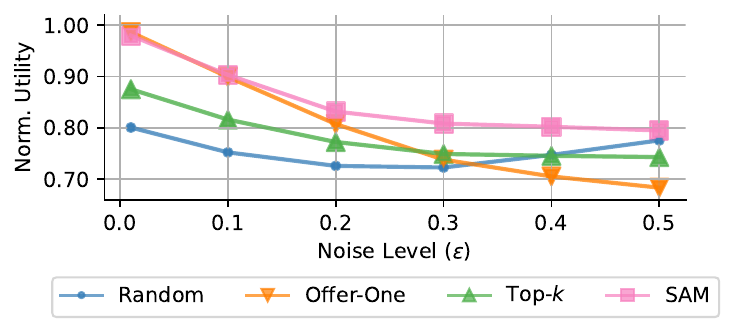} 
    \caption{Performance under the MNL choice model with varied noise level. The noise level plot under a stochastic choice model matches that of the deterministic model, with Offer-One matching SAM at small $\epsilon$, but SAM outperforming Offer-One for larger noise levels.}
    \label{fig:utility_vs_noise_mnl}
\end{figure*}

\textit{Choice Model:} Our main experiments assume that patients deterministically select their highest-utility available provider. 
We test whether our findings are robust to stochastic patient behavior by instead modeling patient choice using a multinomial logit (MNL) model. 
Under MNL choice, patients select higher-utility providers with greater probability but may not always select their utility-maximizing option.
We hold the MNL choice parameters fixed and vary the preference uncertainty $\epsilon$ as in our primary experiment.
Patients select from available set with probability proportional to $\exp(\theta_{ij}/\tau)$, with $\tau = 0.1$. 
We compute normalized utility relative to deterministic optimal algorithm with true knowledge of the underlying preferences, rather than the optimal MNL menu. 
Figure~\ref{fig:utility_vs_noise_mnl} demonstrates that the performance gap between SAM and Offer-One increases as the noise level $\epsilon$ increases. 
This matches the result from Figure~\ref{fig:noise_sweep}, with choice having more value with larger noise. 
Thus, the benefits of tailored choice under uncertainty do not depend on a deterministic selection rule. 

\begin{figure*}
\centering
\includegraphics[width=\textwidth]{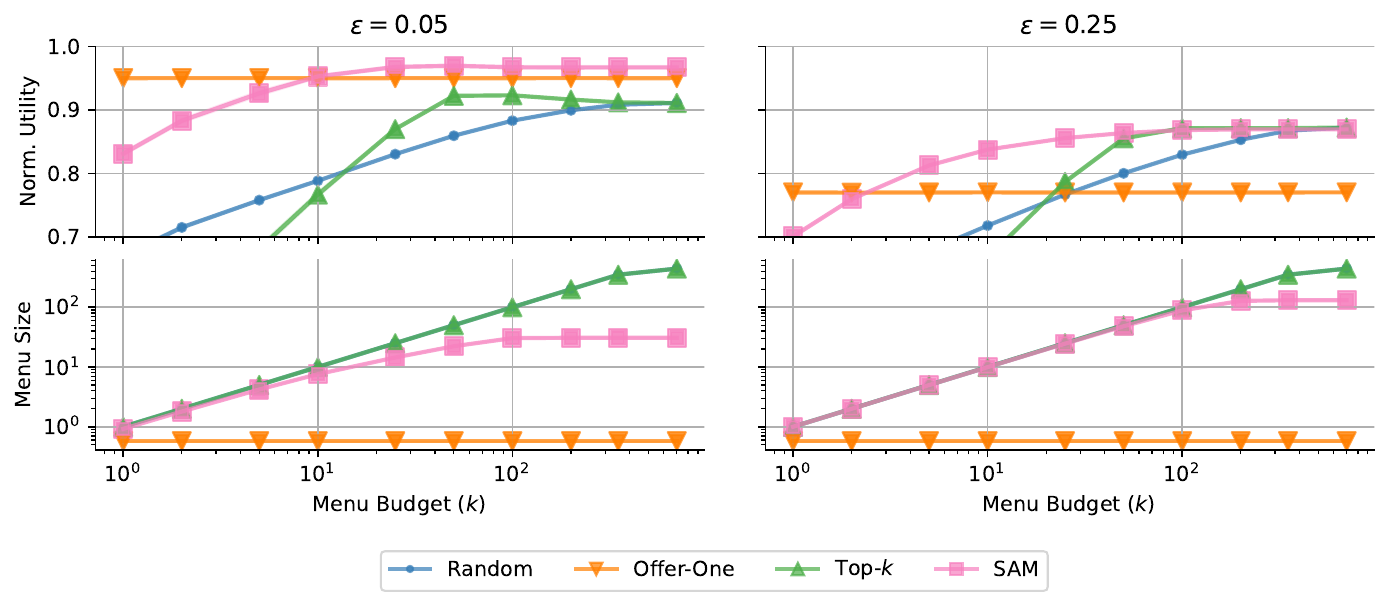}
\caption{Impact of menu budget $k$ on the normalized utility (top) and menu size (bottom) for $\epsilon=0.05$ and $\epsilon=0.25$.
Increasing $k$ leads to better performance for SAM as SAM benefits from additional menu flexibility at both noise levels, while unrestricted choice is valuable primarily under greater preference uncertainty. }
\label{fig:utility_vs_menu_budget}
\end{figure*}

\textit{Menu Budget:} We characterize the impact of the menu budget $k$ on performance in Figure~\ref{fig:utility_vs_menu_budget}. 
Here, $k$ captures real-world constraints on the size of a patient's consideration set, such as limited time, attention, or capacity to evaluate multiple providers. 
In the top panel, we vary menu budget $k$ and measure the impact on the normalized utility. 
For both levels of noise, SAM improves as the menu budget increases. 
At $k=1$, SAM underperforms Offer-One because providers can be over-offered relative to Offer-One's pure matching. As $k$ grows, SAM's menus include fallback options instead of forcing exits, and patients can act on their private preference information; SAM overtakes Offer-One at $k \approx 5$ when $\epsilon = 0.25$ and $k \approx 10$ when $\epsilon = 0.05$. 
At $\epsilon=0.05$, Top-$k$ and Random perform worse than Offer-One at every budget, and Top-$k$'s performance degrades after a certain point with increasing $k$ as it approaches the fully unconstrained Offer-All ($k=M$). 
At $\epsilon=0.25$, larger $k$ consistently improves both baselines, which outperform Offer-One for large enough $k$ and eventually converge to SAM with sufficient budget. 
In the bottom panel, we show that increasing $k$ increases menu sizes, but that SAM plateaus after some point as the menu limit no longer binds. 

\section{Discussion}
\label{sec:discussion}
Our analysis yields three managerial insights to guide the construction of patient--provider matching algorithms in practice: 

\textbf{Insight \#1: Offering limited choice is an operational lever for managing uncertainty.} 
Patient choice allows matching systems to exploit preference information that is unavailable when menus are constructed.
Rather than committing to a single provider based on noisy estimates, offering a tailored set of alternatives allows patients to resolve uncertainty using their realized preferences.
Choice should be viewed not simply as a patient-centered feature, but as an operational mechanism for hedging against estimation error.

\textbf{Insight \#2: The value of choice depends on both preference uncertainty and preference spread.} Patient choice has the largest impact on performance when preference estimates are uncertain and when options differ enough for the ranking to matter.
In these settings, the platform has difficulty distinguishing among alternatives, while patients can use their private preference information to do so. Platforms should focus on expanding choice in these situations. 
Conversely, when one provider is clearly preferred or preference estimates are precise, additional choice provides less informational value, and when providers are largely interchangeable, there is little utility differentiation to reveal. In these cases, choice can instead create unnecessary competition for capacity.

\textbf{Insight \#3: Capacity-aware choice can improve geographic equity by diffusing access to scarce providers.} In our semi-synthetic experiments, tailored choice reduces the urban--rural match quality gap. 
More broadly, this suggests that menu design can affect not only aggregate welfare but also the distribution of high-quality matches across patients. 
Whether similar improvements extend to other dimensions of equity depends on the underlying preference and capacity structure and remains an important direction for future work.

\paragraph{Beyond patient--provider matching:} 
Many capacity-constrained platforms face the same two challenges: hedging against uncertain preferences and coordinating user choices against scarce capacity under unknown response order. Online marketplaces, rideshare and delivery platforms, and appointment scheduling all share this structure~\citep{food_delivery,appointment_booking}. 
Common practice in these settings often resembles our baselines, such as open directories that let everyone choose freely (Offer-All) or direct assignment that removes choice (Offer-One). 
Our results suggest that tailoring choice can be a deliberate lever to hedge against preference uncertainty and manage selection-time congestion.

\paragraph{SAM in practice:}
In practice, $\boldsymbol{\hat{\theta}}$ can be computed by training predictive models on historical data. Inputs can include patient and provider features such as comorbidities, specialties, insurance, and location, together with realized matches and outcomes like satisfaction or operational proxies such as adherence and provider switches. Because historical matches reflect non-random assignment, causal inference methods can be applied to account for selection bias. These estimates can be augmented with patient surveys or discrete choice experiments, which recover preference weights directly.

After computing $\bthetahat$, practitioners can estimate a distribution over plausible preferences through bootstrap sampling. 
Specifically, they can repeatedly resample the historical data, refit the preference model, and generate a predicted utility matrix for each fitted model. 
The resulting collection of utility matrices provides an empirical approximation to preference uncertainty that SAM can use to sample plausible preference realizations. 
This procedure need not recover the full distribution of private patient preference information, but provides one practical way to represent uncertainty when such information is not directly observed.

Finally, in practice, SAM can be dynamically re-run as providers enter and leave the platform. To preserve the early-response incentive established in Proposition~\ref{prop:early-response} under repeated re-runs, the platform can restrict rematching in one of two ways: previously matched patients are revisited only when their own provider exits, rather than being reopened every time a new round runs; or newly arriving patients are held back for a short delay before being pooled with previously matched patients in the next round.

\paragraph{Limitations:}
Our computational experiments are performed using a semi-synthetic dataset, which is used primarily to understand the impact of different operational parameters. 
However, further analysis using real-world patient preference data is necessary to further characterize SAM's performance in practice. 
Factors beyond the set of providers offered might also play a role, such as the ordering of options and the user interface for selecting providers, which should be accounted for in real-world implementation. 
Our objective captures the value of patient choice through better patient--provider matches, but does not separately account for patients' preferences for having more options.
Finally, while we study fairness via a geographical lens, other dimensions also warrant study, such as across race or gender. 
We note that these dimensions require further assumptions on patient utility functions, which we leave to future work. 

\paragraph{Extension to provider workloads and preferences:}
We discuss two extensions to provider workloads and provider preferences.
The former refers to different patients requiring different amounts of effort from providers. 
In essence, we can model this as patients having heterogeneous capacity consumption, then modify the optimization problem appropriately. 
The latter refers to two-sided utility, where providers might have preferences over case complexity, panel size, or the patients they see. 
The key challenge is the asymmetry between patient choice and provider utility, as patients optimize for their own utility while ignoring provider utility.  
An extension to two-sided utility could incorporate a new utility matrix $\phi$, along with a constraint that offered patients achieve some minimum provider utility. Further incorporating provider utility into the objective would require a reformulation of the problem to handle two-sided choice, which could serve as a direction for future work. 

\section{Conclusion}
\label{sec:conclusion}
Choice allows patients to act on private preference information that is unavailable to the health system, but left unrestricted, it can create competition for scarce provider capacity.
We formalize this tension through menu and coordination loss, showing that neither direct assignment nor unrestricted choice is uniformly optimal for all problem instances.
SAM navigates this tradeoff by constructing capacity-aware menus that preserve valuable alternatives while limiting competition over scarce providers.
Our theoretical and empirical results demonstrate that tailored menus can improve welfare under preference uncertainty and, in our healthcare setting, can also reduce geographic disparities in match quality.
We find that choice becomes more valuable as preference estimates grow more uncertain and as provider capacity grows, while tailoring it matters most when capacity is scarce.
More broadly, our results suggest that the amount of choice offered to users should itself be treated as an operational decision.
\section*{Acknowledgements}
We thank Aakash Lahoti, Keegan Harris, Kevin Li, Ariana Tang, Fei Fang, Rema Padman, Erica Chiang, Tori Qiu, and George Chen for their comments. We thank the reviewers at ML4H for their feedback on a preliminary version of this work, along with the audiences at INFORMS and MSOM. 
We thank Susan Barrett and Karen Hall for their input from a health system perspective. 
Co-author Raman is supported by an NSF GRFP Fellowship and a K\&L Gates Fellowship.

\newpage 

\bibliographystyle{pomsref}
\bibliography{references}

\newpage 

\appendix 
\section{Experimental Details}
\label{app:experimental_details}
We set the patient and provider sizes to $N=1225$ and $M=700$ from the average panel size and provider count at our partner organization.
For the response order, we let $\sigma \sim \mathcal{D}$, where $\mathcal{D}$ is the uniform distribution over permutations. 
We construct $\boldsymbol{\hat\theta}$ to reflect geographic proximity and the presence of comorbidities. 
Formally, for a patient $i$ and provider $j$, let $d_{ij}$ be the distance between the patient and provider, and let $\beta_{ij}$ denote whether patient $i$'s comorbidity and $j$'s specialty match.  
For example, if patient $i$ has a heart-related comorbidity and provider $j$ has a cardiology specialty, then $\beta_{ij}=1$. 
The match quality is $\hat{\theta}_{ij} = \alpha + (1-\alpha)\big(\omega \beta_{ij} + (1-\omega)\mathrm{Clip}(\frac{\bar{d}}{d_{ij}}-1,0,1)\big)$, where $\omega$ and $\alpha$ are hyperparameters that modulate the impact of distance and comorbidities. 
Here, $\alpha$ sets the lower limit on the quality of scores, while $\omega$ controls the relative weight placed on comorbidity match versus geographic proximity. 
By default, we let $\omega=\alpha=\frac{1}{2}$. 
We use these coefficients to compute $\hat{\theta}$, which is the estimated preferences rather than the true preferences. 

We compute $d_{ij}$ by sampling patient locations from Connecticut ZIP codes and obtaining provider locations from a Medicare dataset~\citep{medicare} restricted to Connecticut providers listing internal medicine, general practice, or family practice (1,098 providers), sampled with replacement. 
We construct $\beta_{ij}$ using prior work on comorbidity rates~\citep{age_comorbidity} and information on provider specialties from a Medicare dataset~\citep{medicare}. 
Here, $\beta_{ij}=1$ if patient $i$ and provider $j$ overlap in at least one of five condition categories.
Finally, we set $\bar{d}=20.2$ since prior work demonstrates the average distance threshold for patients is 20.2 miles~\citep{washington_distance}. 
We set the utility for not matching as $\theta_{i,M+1} = 0.25$ and let the provider capacities be $\bar{c}_{j} \sim \mathrm{Poisson}(1)$. 

\section{Sample-Average Optimal Solution via MILP}
\label{app:milp}

We demonstrate that the sample-average optimal menu can be computed by solving a MILP. 
We use this in Section~\ref{subsec:when_is_choice_valuable} to compute approximation ratios for different policies. 
Let $y_{ij,s}$ represent a binary variable indicating whether patient $i$ selects provider $j$ in scenario $s$, where each scenario contains a realization of both patient preferences and response order.
Using this notation, we can find the sample--average optimal menu $\mathbf{X}$ over the sampled scenarios by solving the following optimization problem. 

\begin{equation}
\label{eq:milp}
\begin{aligned}
\max\limits_{\{X_{ij}, y_{ij,s}\}} \quad 
& \frac{1}{S}\sum_{s=1}^{S} \sum_{i=1}^{N} \sum_{j=1}^{M+1} 
y_{ij,s} \theta^{(s)}_{ij} \\
\text{subject to} \quad 
& \sum_{j=1}^{M} X_{ij} \le k 
&& \forall i \in [N] \\
& \sum_{j=1}^{M+1} y_{ij,s} = 1 
&& \forall i \in [N], s \in [S] \\
& y_{ij,s} \le X_{ij} 
&& \forall i \in [N], j \in [M], s \in [S] \\
& \sum_{i=1}^{N} y_{ij,s} \le \bar{c}_{j} 
&& \forall j \in [M], s \in [S] \\
& c_{t,j,s} = \bar{c}_j - \sum_{t' < t} 
y_{\sigma_{t'}^{(s)},j,s}
&& \forall t \in [N], j \in [M], s \in [S] \\
& \sum_{j'=1}^{M+1} 
y_{\sigma_t^{(s)},j',s}
\theta^{(s)}_{\sigma_t^{(s)},j'}
\geq 
\theta^{(s)}_{\sigma_t^{(s)},j}
\mathbf{1}[c_{t,j,s}>0]
X_{\sigma_t^{(s)},j}
&& \forall t \in [N], j \in [M], s \in [S] \\
& \sum_{j'=1}^{M+1} 
y_{\sigma_t^{(s)},j',s}
\theta^{(s)}_{\sigma_t^{(s)},j'}
\geq
\theta^{(s)}_{\sigma_t^{(s)},M+1}
&& \forall t \in [N], s \in [S] \\
& X_{ij}, y_{ij,s} \in \{0,1\}, 
\quad c_{t,j,s} \ge 0.
\end{aligned}
\end{equation}

The first constraint limits each patient's menu to at most $k$ providers, while the next three constraints ensure that each patient selects exactly one option, can only select an offered provider, and respects provider capacity. 
The remaining constraints enforce sequential patient choice: when patient $\sigma_t^{(s)}$ responds, they select the highest-utility option among the providers in their menu that still have available capacity and the exit option.

The capacity indicator $\mathbf{1}[c_{t,j,s}>0]$ and its product with
$X_{\sigma_t^{(s)},j}$ can be linearized exactly using standard auxiliary binary variables, so Equation~\ref{eq:milp} admits an exact MILP formulation.
Solving this MILP is computationally expensive, which is why we only compute approximation ratios when $N$ and $M$ are small.
\section{Ablations}
\label{app:ablations}

\subsection{Distribution of $\hat{\theta}$}
\label{app:other_theta}
\begin{figure*}[]
    \centering 
\includegraphics[width=\textwidth]{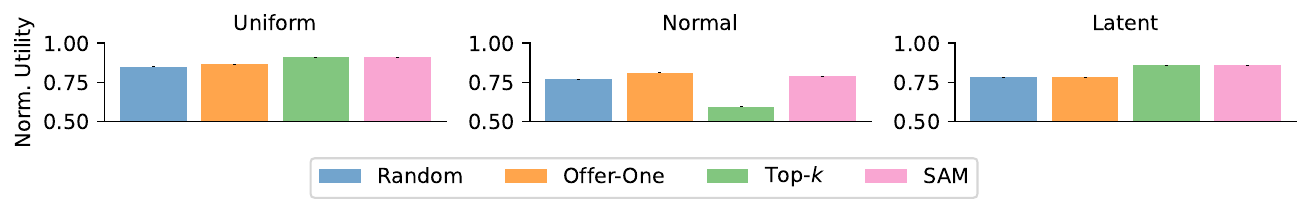} 
    \caption{Performance when varying the distribution for $\boldsymbol{\hat{\theta}}$ between uniform ($\hat{\theta}_{i,j} \sim U(0,1)$), normal ($\hat{\theta}_{i,j} \sim \mathcal{N}(\mu_j,\sigma^2)$), and latent ($\hat{\theta}_{i,j} \sim u_{i} \cdot v_{j} + \mathcal{N}(0,\sigma^2)$). Offer-One performs best for $\boldsymbol{\hat{\theta}}$ normal, while Top-$k$ and SAM perform best for uniform and latent.}
    \label{fig:utility_function_ordering}
\end{figure*}

In Figure~\ref{fig:utility_function_ordering}, we compare performance when varying $\boldsymbol{\hat{\theta}}$ as uniform ($\hat{\theta}_{i,j} \sim \mathrm{U}(0,1)$), normal ($\hat{\theta}_{i,j} \sim \mathcal{N}$ and  $(\mu_{j},\sigma^2), \mu_{j} \sim U(0,1)$), and latent vectors ($\hat{\theta}_{i,j} \sim u_{i} \cdot v_{j} + \mathcal{N}(0,\sigma^2)$), with $\sigma=0.05$.
When $\hat{\boldsymbol{\theta}}$ is distributed according to a uniform or latent distribution, SAM performs best, while for a normal distribution, Offer-One performs best.
Under uniform and latent distributions, patients' preferences are heterogeneous, reducing coordination loss, while for normal distributions, preferences are homogeneous, which can increase inter-patient competition and coordination loss. 

\subsection{Operationalizations of Utility}
\label{app:utility_operationalization}
\begin{figure*}
    \centering 
    \includegraphics[width=\textwidth]{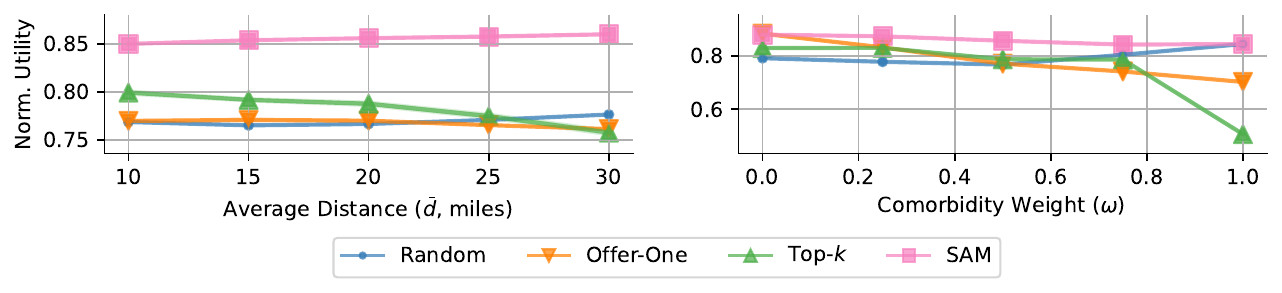}
    \caption{We vary the parameters underlying our operationalization of $\theta_{i,j}$: $\bar{d}$, which controls the sensitivity of utility to distance, and $\omega$, which controls the weight placed on comorbidity match. We find that across values of $\bar{d}$ and $\omega$, SAM performs best, showing that our results are robust to our selection of hyperparameters.}
    \label{fig:utility_operationalization}
\end{figure*}

In Figure~\ref{fig:utility_operationalization}, we compare the performance of algorithms when we vary two parameters impacting our construction of $\hat{\theta}_{i,j}$: $\bar{d}$ and $\omega$.
The former controls sensitivity to distance, with smaller \(\bar d\) indicating greater sensitivity, while larger \(\omega\) places greater weight on comorbidity match relative to distance. 
We show that, across values of both, SAM performs best, showing that our empirical results are robust to our exact operationalization of $\hat{\theta}_{i,j}$. 

\subsection{Stronger Baselines}
\label{app:stronger_baselines}
\begin{figure*}[]
    \centering 
\includegraphics[width=\textwidth]{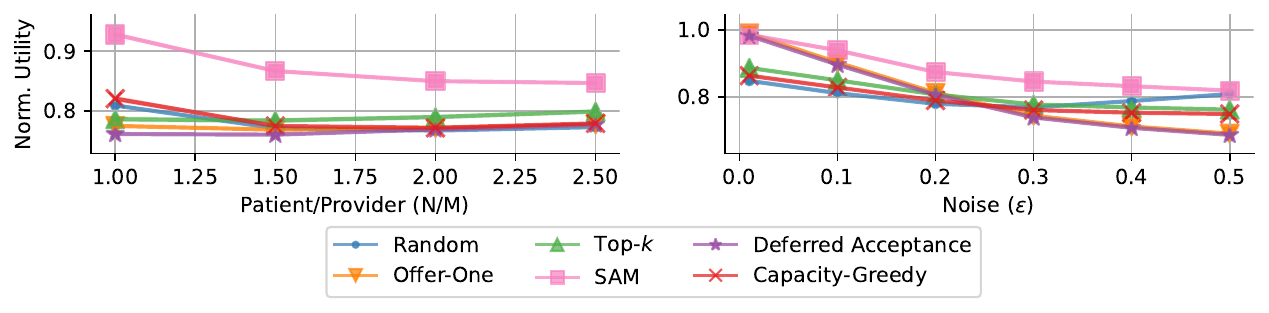} 
    \caption{Comparison of two new baselines, deferred acceptance and capacity-greedy, when varying the patient--provider ratio (left) and noise $\epsilon$ (right). SAM outperforms both baselines for all patient--provider ratios, and for $\epsilon \geq 0.1$.}
    \label{fig:stronger_baselines}
\end{figure*}

To understand whether stronger baselines improve performance, we compare against two additional baselines: deferred acceptance and a capacity-greedy algorithm. 
The first uses the deferred acceptance framework where patients iteratively propose to providers, and provider utilities correspond to $\theta_{\cdot,j}$, resulting in a set of pairs between patients and providers~\citep{stable_matching}. 
The second processes patients in decreasing order of their best provider's utility, offering each patient the $k$ highest-utility providers with remaining capacity and tentatively assigning them to the best of those, which decrements that provider's capacity.
In Figure~\ref{fig:stronger_baselines}, we show that SAM outperforms both of these baselines for $\epsilon \geq 0.1$ and all $N/M$; the former closely tracks Offer-One while the latter tracks Top-$k$.  
These results show that SAM's performance gains persist relative to stronger matching and capacity-aware baselines. 

\subsection{Decomposing Menu and Coordination Loss}
\label{app:dual-assumptions}
\begin{figure*}[]
    \centering 
\includegraphics[width=\textwidth]{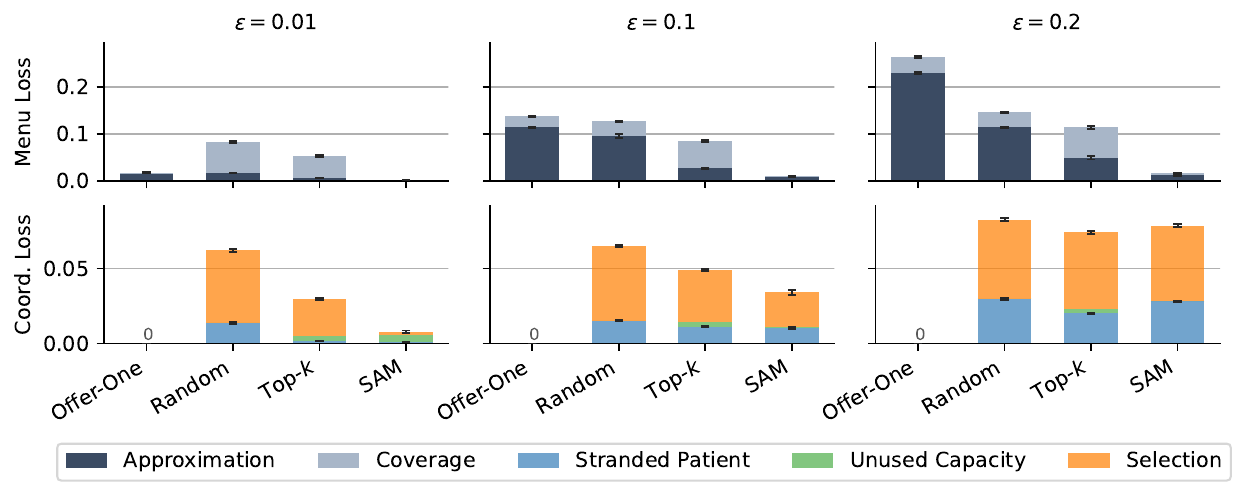} 
    \caption{Top: the menu-loss bound from Lemma~\ref{lem:menu_loss}, decomposed into approximation loss and coverage loss. Bottom: the exact coordination-loss decomposition from Lemma~\ref{lem:coordination_loss} into stranded patient, unused capacity loss, and selection slack. For SAM, the dominant realized coordination-loss component is selection slack.}
    \label{fig:more_error_decomposition}
\end{figure*}

To better understand how error decomposes across policies, we decompose the menu-loss bound into approximation loss ($tC$), which captures the loss from using edges within $t$ of optimal, and coverage loss ($\Def_t(\mathbf{X})$), which captures the loss when the menu does not contain enough such edges.
We decompose coordination loss into three terms corresponding to Lemma~\ref{lem:coordination_loss}: stranded patient ($\sum_{i\notin F}\mu_i^X$), arising from patients who remain unmatched; unused capacity loss ($\sum_j\lambda_j^X(\bar c_j-f_j)$), arising from valuable provider capacity that remains unfilled; and selection slack ($\sum_{i\in F}g^X_{i,s_i}$), arising when patients select edges with positive slack relative to the menu-restricted optimal allocation.
We use the same $N=1225$ and $M=700$ setting from Section~\ref{sec:empirical}, while deriving $\boldsymbol{\hat{\theta}}$ from distance and comorbidity data. 
In Figure~\ref{fig:more_error_decomposition}, for SAM, we find that coverage loss is near zero for all values of $\epsilon$, while Offer-One has the largest coverage loss, reflecting how not offering choice leads to larger menu loss. 
For coordination loss, the dominant component is selection slack, which arises when patients select providers that differ from those chosen by the optimal allocation supported by the menu.
Empirically, SAM's primary source of loss is not a failure to construct menus containing good allocations, but the coordination cost introduced when patients exercise choice within those menus.
This decomposition illustrates the central tradeoff in menu design: broader menus can reduce menu loss by preserving valuable alternatives, but doing so can increase coordination loss by allowing decentralized choices to deviate from the system-optimal allocation.
In practice, this decomposition can help diagnose whether a policy's losses arise primarily from insufficient menu coverage or from coordination under patient choice. 

\subsection{Response Order and Noise Distribution}
\label{app:response_order_noise_distribution}
\begin{figure*}
    \centering 
    \includegraphics[width=\textwidth]{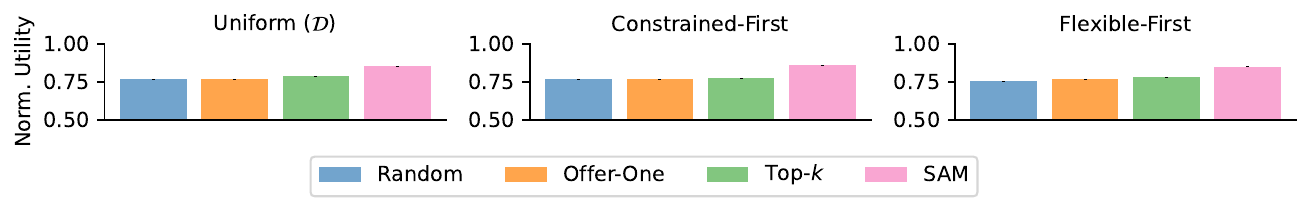}
    \includegraphics[width=\textwidth]{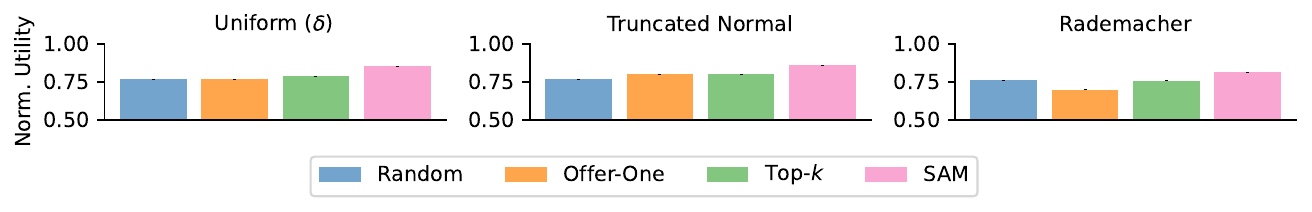}
    \caption{We vary the arrival order distribution $\mathcal{D}$ (top) and the noise distribution for $\delta_{i,j}$ (bottom). In both situations, we find that there is little impact on performance of different policies, with SAM performing best across choices of arrival order distribution and noise distribution.}
    \label{fig:utility_arrival_noise_distro}
\end{figure*}

We additionally test robustness to the arrival-order distribution $\mathcal{D}$ and the distribution of preference noise $\delta_{ij}$ (Figure~\ref{fig:utility_arrival_noise_distro}).
For arrival order, we compare our baseline uniform random ordering against two preference-dependent orderings.
For each patient, let $q_i(\epsilon)$ denote the number of providers whose estimated utility is within $\epsilon$ of the patient's highest estimated utility.
In the \emph{constrained-first} ordering, patients are ordered from smallest to largest $q_i(\epsilon)$, so patients with fewer viable alternatives respond first; in the \emph{flexible-first} ordering, this order is reversed.
For preference noise, we hold the support $\delta_{ij}\in[-\epsilon,\epsilon]$ fixed and compare our baseline uniform distribution against Rademacher noise, which places equal probability on $\{-\epsilon,+\epsilon\}$, and a centered normal distribution truncated to $[-\epsilon,\epsilon]$.
Across both ablations, the relative performance of the policies is largely unchanged, with SAM performing best under each arrival-order and noise specification.

\subsection{Number of Samples}
\label{app:num_samples}
\begin{figure*}[]
    \centering 
\includegraphics[width=0.5\textwidth]{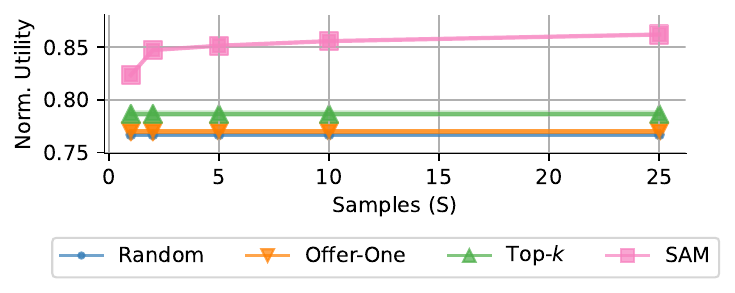} 
    \caption{Performance of SAM as the number of samples \(S\) increases. Performance quickly plateaus; \(S=10\) achieves 99.3\% of the normalized utility obtained with $S=25$.}
\label{fig:utility_vs_num_samples}
\end{figure*}

In Figure~\ref{fig:utility_vs_num_samples}, we vary the number of samples $S \in \{1,2,5,10,25\}$ and compute the performance of SAM. 
We find that performance quickly plateaus with increasing $S$, with $S=10$ capturing 99.3\% of the performance of $S=25$. 
Performance plateaus because a small number of samples is sufficient to accurately estimate the scenario-averaged shadow prices $\hat{\lambda}_{j}$. 
This suggests that relatively few samples suffice to estimate the shadow prices needed for good menu construction. 

\subsection{Runtime}
\label{sec:runtime}
\begin{figure*}[]
    \centering 
\includegraphics[width=\textwidth]{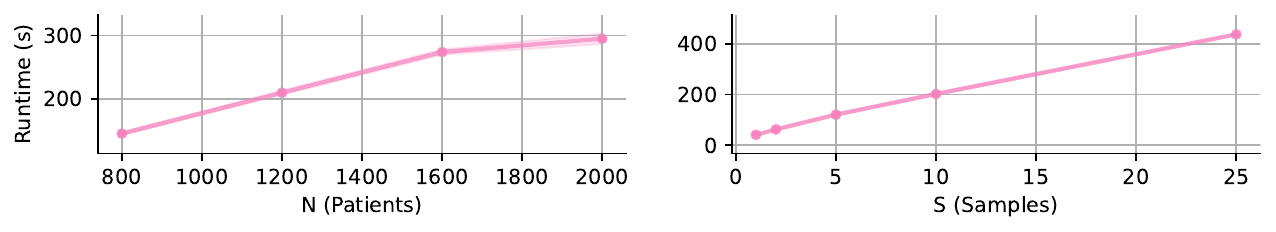} 
    \caption{Runtime for SAM as the number of patients ($N$) and samples ($S$) increase. Runtime increases approximately linearly in both \(N\) and \(S\), remaining below 500 seconds across the settings considered. }
\label{fig:runtime}
\end{figure*}

We explore the runtime of SAM in Figure~\ref{fig:runtime} when varying the number of patients $N$ and the number of samples $S$. 
Runtime increases approximately linearly with both parameters, with the algorithm running in under 500 seconds in all scenarios (on a computer with 48 GB RAM). 
Even at real-world problem scales, SAM runs efficiently because its most expensive part, computing $S$ linear programs, can be done quickly. 
These runtimes suggest that SAM is computationally feasible for batch reassignment at the problem scales we consider.
This runtime is compatible with batch reassignment settings, where organizations periodically rematch patient panels after provider exits or scheduling changes. 
As a result, practitioners can use SAM as a decision-support tool without introducing substantial delays into existing care coordination workflows. 
\section{Early Response Proof}
\label{app:early_response_proof}
\begin{proof}[Proof of Proposition~\ref{prop:early-response}]
Fix any ordering $\sigma$ and, for a position $t \in [N]$, let
 $\sigma_{i:t}$ denote the ordering that places patient $i$ at position
$t$ and keeps the remaining patients in their relative order from
$\sigma$. Define $U_i(t;\sigma,\mathbf{X})$ as patient $i$'s realized
utility under $\sigma_{i:t}$; assuming a fixed tie-breaking rule, this
quantity is deterministic. For any provider, let
$g_j(\mathbf{X},\sigma,t) = \mathbf{1}[c_{t,j}(\mathbf{X},\sigma) > 0]$
denote whether provider $j$ has remaining capacity at time $t$ along
the trajectory of $\sigma$, with
$g_{M+1}(\mathbf{X},\sigma,t) = 1$ for the exit option.
 
Fix $t \le t'$. The capacity state entering position $t$ coincides
under $\sigma_{i:t}$ and $\sigma_{i:t'}$: positions $1,\dots,t-1$ are
occupied by the same patients in the same order, and selection is a
deterministic function of the menu, the true utilities, and the
remaining capacities, so by induction on positions their selections
coincide. Patient $i$'s feasible set under $\sigma_{i:t}$ is therefore
$\{\, j : g_j(\mathbf{X},\sigma_{i:t'},t)=1,\ X_{ij}=1 \,\}
\cup \{M{+}1\}$, evaluated along the trajectory of $\sigma_{i:t'}$.
Along that single trajectory, provider availability is monotonically
decreasing, $g_j(\mathbf{X},\sigma_{i:t'},t) \ge
g_j(\mathbf{X},\sigma_{i:t'},t')$, so the feasible set patient $i$
faces at position $t'$ is contained in the one at position $t$.
Because $U_i$ is a maximum of $\theta_{ij}$ over the feasible set,
which contains the exit option in both cases,
\[
U_i(t;\sigma,\mathbf{X}) \;\ge\; U_i(t';\sigma,\mathbf{X})
\qquad \forall\, t \le t'.
\]
The inequality holds pointwise in $\sigma$, so taking expectations over
$\sigma \sim \mathcal{D}$ yields
$\mathbb{E}_{\sigma}[U_i(t;\sigma,\mathbf{X})] \ge
 \mathbb{E}_{\sigma}[U_i(t';\sigma,\mathbf{X})]$ for all $t < t'$,
for any distribution $\mathcal{D}$. 
\end{proof}

\section{Menu Design Limit Proofs}
\label{app:menu_design_proofs}

\subsection{Computational Intractability}
\label{app:intractability_proof}

\begin{proof}[Proof of Theorem~\ref{thm:nphard}]
We prove NP-hardness by reduction from the \textit{Product Partition} problem, which is NP-hard \citep{product_partition}. An instance of \textit{Product Partition} is a set of weights $w_1,\dots,w_N\in\mathbb{N}$ with $w_i\ge 1$, and the goal is to find disjoint sets $S_1,S_2$ with $S_1\cup S_2=[N]$ so $\prod_{i \in S_1}^{N} w_{i} = \prod_{i \in S_{2}} w_{i}$.

We reduce this problem to the problem of finding an optimal menu design, and in particular, of finding an optimal menu $\mathbf{X}$. Given weights $w_i$, construct the following menu design instance. Let $M=2$, $k=1$, $\epsilon=1$, $\bar c_1=\bar c_2=1$, and $\mathcal{D}$ uniform over $[N]$. Set the exit utility $\theta_{i,3}=\gamma$ for all $i$, with $0<\gamma<1$, observed without noise. For providers $j\in\{1,2\}$, let the true utilities be drawn independently with
\[
\theta_{ij} =
\begin{cases}
1 & \text{with probability } q_i := 1-1/w_i,\\
0 & \text{otherwise,}
\end{cases}
\]
so $\mathcal{P}_{\boldsymbol\theta}$ is the product of these Bernoulli marginals. The platform observes the unbiased estimate $\hat\theta_{ij}=q_i$ for $j\in\{1,2\}$. The induced noise $\delta_{ij}=\theta_{ij}-\hat\theta_{ij}\in\{-q_i,\,1-q_i\}$ has mean zero,
\[
\mathbb{E}[\delta_{ij}]=\mathbb{E}[\theta_{ij}]-q_i=q_i-q_i=0,
\]
and satisfies $\|\boldsymbol\delta\|_\infty=\max_i\max(q_i,1- q_i)\le\epsilon=1$, with $\theta_{ij}\in\{0,1\}\subseteq[0,1]$. 

Since $k=1$, a menu assigns each patient to at most one provider and is described by disjoint sets $S_1,S_2\subseteq[N]$, where $S_j$ is the set offered provider $j$. Provider $j$ is matched whenever some patient in $S_j$ realizes $\theta_{ij}=1$, i.e., with probability $1-\frac{1}{\prod _{i \in S_j} w_i}$, which is increasing in $S_j$. 
Since adding a patient to either menu weakly increases the probability that the provider is matched and does not affect the other provider, there exists an optimal solution in which every patient is assigned to one of the two providers.

Let $W_j:=\prod_{i\in S_j}w_i$, so $W_1 W_2=W:=\prod_i w_i$ is fixed. A matched patient yields utility $1$ and an unmatched patient yields $\gamma$, and the expected number of matches is $(1-1/W_1)+(1-1/W_2)$, so the expected total utility (independent of arrival order) is
\[
(1-\gamma)\bigl[\,2-\frac{1}{W_1}-\frac{1}{W_2}\bigr]+N\gamma.
\]
Since $1-\gamma>0$ and $N\gamma$ is constant, this is maximized by minimizing $\frac{1}{W_1} + \frac{1}{W_2}$ at fixed $W_1*W_2$, which by the AM-GM inequality occurs at the partition $W_1 = W_2$. Solving menu optimization thus solves the \textit{Product Partition} problem, so the problem is NP-hard, even under $k=1$, $M=2$, unit capacities, and zero-mean Bernoulli noise.

\end{proof}

\subsection{Information Theoretic Lower Bound}
\label{app:info_theoretic_proof}
\begin{proof}[Proof of Theorem~\ref{thm:info_theory_lower_bound}]
    \textbf{Instance Construction.} 
    For capacities such that $\alpha \bar c_j$ is integral, partition the $N$ patients into $M$ disjoint groups $G_1,\ldots,G_M$ with $|G_j|=\alpha\bar c_j$.    
    Fix a small constant $\gamma$ with $0 < \gamma \ll \frac{1}{2} - \epsilon$, and set the exit utility $\theta_{i,M+1} = \gamma$ for all patients $i$. 
    Set $\theta_{ij} = 0$ for all $i \notin G_j$; patients outside group $j$ receive zero utility from provider $j$ and will always prefer the exit option over provider $j$, since $0 < \gamma$. 
    Set 
    \[
        \hat{\theta}_{ij} = \frac{1}{2} \quad \text{for all } i \in G_j, \quad
        \hat{\theta}_{ij} = 0 \quad \text{for all } i \notin G_j.
    \]
    Additionally, let $\hat{\theta}_{i,M+1} = \gamma$. 
    Next, let 
    \[
        \delta_{ij} =
        \begin{cases}
            +\epsilon & \text{with probability } \tfrac{1}{2} \\
            -\epsilon & \text{with probability } \tfrac{1}{2}
        \end{cases}
        \quad \text{for all } i \in G_j.
    \]   
    We note that $\delta_{ij} = \theta_{ij} -\hat{\theta}_{ij}\in [-\epsilon,\epsilon]$ and is zero mean. 

    \textbf{Omniscient welfare.}    
    The optimal omniscient algorithm (with knowledge of $\theta$) offers provider $j$ to the top-$\bar{c}_{j}$ patients with the largest values of $\delta_{ij}$. 
    This is equivalent to
\[
    OPT = \sum_{j=1}^{M} (\alpha-1) \gamma \bar{c}_{j}+  \sum_{j=1}^{M} \left[ \bar{c}_j \cdot \frac{1}{2} + \sum_{i \in \text{top-}\bar{c}_j(G_j)} \delta_{ij} \right]
    = (\alpha-1) \gamma \sum_{j=1}^{M} \bar{c}_{j} + \frac{1}{2}\sum_{j=1}^M \bar c_j \;+\; \sum_{j=1}^M \sum_{i \in \text{top-}\bar{c}_j(G_j)} \delta_{ij},
\]
where $\text{top-}\bar{c}_j(G_j)$ denotes the $\bar{c}_j$ patients in $G_j$ with the largest $\delta_{ij}$ (since $\hat{\theta}_{ij}=\frac{1}{2}$ is constant within $G_j$, ranking by $\hat\theta_{ij}+\delta_{ij}$ within the group is the same as ranking by $\delta_{ij}$ alone).

The values $\{\delta_{ij}\}_{i \in G_j}$ are $|G_j| = \alpha \bar{c}_j$ i.i.d.\ $\pm\epsilon$ coin flips. Let $X_j$ denote the number of $i \in G_j$ with $\delta_{ij} = +\epsilon$; then $X_j \sim \text{Binomial}(\alpha \bar{c}_j, \tfrac12)$, with mean $\tfrac{\alpha \bar c_j}{2}$. 
Concentration then implies
\[
X_j=\left(\frac{\alpha}{2}+o(1)\right)\bar c_j.
\]
Since the top $\bar c_j$ patients contain
$\min\{X_j,\bar c_j\}$ patients with perturbation $+\epsilon$
and $\bar c_j-\min\{X_j,\bar c_j\}$ patients with perturbation
$-\epsilon$,
\[
\sum_{i \in \text{top-}\bar c_j(G_j)} \delta_{ij}
=
\epsilon\left(2\min\{X_j,\bar c_j\}-\bar c_j\right).
\]
Using
\[
X_j=\left(\frac{\alpha}{2}+o(1)\right)\bar c_j,
\]
we obtain
\[
\sum_{i \in \text{top-}\bar c_j(G_j)} \delta_{ij}
=
\left(\min\{1,\alpha-1\}\epsilon-o(1)\right)\bar c_j.
\]
Relating back to the optimal solution, we have that 
\[
\mathbb{E}[OPT]
=
\left(
(\alpha-1)\gamma+\frac12
+\min\{1,\alpha-1\}\epsilon-o(1)
\right)
\sum_{j=1}^M \bar c_j.
\]

\textbf{Algorithm welfare and the gap.}
Fix provider $j$. Any patient outside $G_j$ has zero utility for provider $j$ and never consumes a slot, so we may restrict attention to patients in $G_j$.

Since $\frac12-\epsilon>\gamma$, every patient in $G_j$ strictly prefers provider $j$ to the exit option, so any algorithm fills at most $\bar c_j$ slots, and any optimal algorithm fills all of them.
The deterministic algorithm constructs its menu using only \(\hat\theta\) (and whatever other information is available before \(\theta\) is realized). 
Because the perturbations \(\delta_{ij}\) are independent of this information and of the response order, the set of patients $S_j$ who ultimately consume provider \(j\)'s capacity is independent of their perturbations.
Since $\mathbb{E}[\delta_{ij}]=0$,
\[
    \mathbb{E}\!\left[\sum_{i\in S_j}\theta_{ij}\right]
    \leq
    \frac12\,\bar c_j.
\]
Therefore,
\[
    \mathbb{E}[ALG]
    \leq
    \left((\alpha-1)\gamma + \frac12\right)
    \sum_{j=1}^M \bar c_j.
\]
Subtracting yields
\[
\mathbb{E}[OPT]-\mathbb{E}[ALG]
\geq
\left(\min\{1,\alpha-1\}\epsilon-o(1)\right)
\sum_{j=1}^M \bar c_j.
\]
\end{proof}
\section{Baseline Algorithm Proofs}
\label{app:baseline_proofs}

\begin{proof}[Proof of Proposition~\ref{prop:robustness_main}]
    We construct a worst-case instance with i.i.d.\ noise.
    For simplicity, let $\bar c_j=1$ for every provider $j$, so that $C=\sum_j \bar c_j=M$, and consider the first $C$ patients.
    Set
    \[
        \hat\theta_{i,i}=\frac12+\gamma
        \qquad\text{and}\qquad
        \hat\theta_{ij}=\frac12
        \quad\text{for all }i\neq j,
    \]
    where $0<\gamma\ll 1/2 - \epsilon$. Let $\theta_{i,M+1} = 0$. 
    Let
    \[
        \theta_{ij}=\hat\theta_{ij}+\delta_{ij},
    \]
    where the noise terms $\{\delta_{ij}\}_{ij\in[C]}$ are i.i.d.\ and satisfy
    \[
        \delta_{ij}
        =
        \begin{cases}
            +\epsilon & \text{with probability } \frac12,\\
            -\epsilon & \text{with probability } \frac12.
        \end{cases}
    \]
    Because $\hat\theta_{i,i}>\hat\theta_{ij}$ for every $j\neq i$, the offer-one policy matches each patient $i$ to provider $i$.
    
    Its expected welfare is therefore
    \[
        \mathbb{E}\!\left[V_\theta(\mathbf X^{\mathrm{One}})\right]
        =
        \sum_{i=1}^C \mathbb{E}[\theta_{i,i}]
        =
        \left(\frac12+\gamma\right)C.
    \]
    
    Now construct a bipartite graph between the $C$ patients and the $C$ providers, placing an edge $(ij)$ whenever $\delta_{ij}=+\epsilon$.
    
    This graph is distributed as the Erd\H{o}s--R\'enyi bipartite graph $G_{C,C}(1/2)$, which contains a perfect matching with probability $1-o(1)$ as $C\to\infty$.    
    On this event, the omniscient policy can select a perfect matching consisting entirely of $+\epsilon$ edges.
    Every edge in this matching has utility at least $\frac12+\epsilon$, and hence
    \[
        OPT
        \ge
        \left(\frac12+\epsilon\right)C
    \]
    with probability $1-o(1)$.
    Consequently,
    \[
        \mathbb{E}[OPT]
        \ge
        \left(\frac12+\epsilon-o(1)\right)C.
    \]
    
    It follows that
    \[
        \mathbb{E}\!\left[OPT-V_\theta(\mathbf X^{\mathrm{One}})\right]
        \ge
        \left(\epsilon-\gamma-o(1)\right)C.
    \]
    Letting $C\to\infty$ and then $\gamma\to0$ gives
    \[
        \mathbb{E}\!\left[OPT-V_\theta(\mathbf X^{\mathrm{One}})\right]
        \ge
        \left(\epsilon-o(1)\right)\sum_j\bar c_j.
    \]

    Next, we show that such a bound is tight up to constant factors. 
    For any matching $\mathbf{X}$, let $A(\mathbf{X}) := \{i : X_{ij} = 1 \text{ for some } j \in [M]\}$ denote the set of patients offered a provider (as opposed to exit).
    Since $\mathbf{X}$ is feasible, the capacity constraints give $|A(\mathbf{X})| = \sum_{j} \sum_{i} X_{ij} \leq \sum_{j} \bar{c}_{j}$, and this holds deterministically for every realization of $\theta$, since it follows only from feasibility of $\mathbf{X}$.
    Clearly, $|A(\mathbf{X})| \leq N$ as well. 
    For any matching $\mathbf{X}$, we may decompose realized welfare as
    \[
    V_{\theta}(\mathbf{X}) = \sum_{i=1}^{N} \theta_{i,M+1} + \sum_{i \in A(\mathbf{X})} (\Delta_{i,s_{i}(\mathbf{X})})_+.
    \]
    When utilities are known, one-provider menus are optimal, so both $\mathbf{X}^{*}$ and $\mathbf{X}^{\mathrm{One}}$ are matchings, and in particular $\mathbf{X}^{*} = \arg\max_{\mathbf{X}} V_{\theta}(\mathbf{X})$ is a matching depending on the realized $\theta$.
    Applying the decomposition to $\mathbf{X}^{*}$ and to its own noisy counterpart gives
    \[
    V_{\theta}(\mathbf{X}^{*}) - V_{\hat{\theta}}(\mathbf{X}^{*}) = \underbrace{\sum_{i=1}^{N} \big(\theta_{i,M+1} - \hat{\theta}_{i,M+1}\big)}_{(\star)} \; + \; \underbrace{\sum_{i \in A(\mathbf{X}^{*})} \big((\Delta_{i,s_{i}(\mathbf{X}^{*})})_+ - (\hat{\Delta}_{i,s_{i}(\mathbf{X}^{*})})_+\big)}_{(\star\star)}.
    \]
    The term $(\star)$ does not depend on $\mathbf{X}^{*}$ at all, so taking expectation over $\theta \sim P_{\theta}$ and using $\mathbb{E}_{\theta}[\theta_{i,M+1}] = \hat{\theta}_{i,M+1}$ from the zero-mean noise assumption gives $\mathbb{E}_{\theta}[(\star)] = 0$ exactly.
    The term $(\star\star)$ satisfies $|(\star\star)| \leq 2\epsilon \, |A(\mathbf{X}^{*})| \leq 2 \epsilon \min(C,N)$ pointwise, for every realization of $\theta$, since each summand is bounded by $2\epsilon$ and $|A(\mathbf{X}^{*})| \leq \min(C,N)$ regardless of which $\theta$ was drawn.
    Combining the two terms and taking expectations yields
    \[
    \mathbb{E}_{\theta}\big[V_{\theta}(\mathbf{X}^{*}) - V_{\hat{\theta}}(\mathbf{X}^{*})\big] \leq 2\epsilon \min(C,N).
    \]
    Then we have 
    \[ 
    \mathbb{E}[V_{\theta}(\mathbf{X}^{*})-V_{\theta}(\mathbf{X}^{\mathrm{One}})] = \mathbb{E}[V_{\theta}(\mathbf{X}^{*})-V_{\hat{\theta}}(\mathbf{X}^{*})] + \mathbb{E}[V_{\hat{\theta}}(\mathbf{X}^{*})-V_{\hat{\theta}}(\mathbf{X}^{\mathrm{One}})] + \mathbb{E}[V_{\hat{\theta}}(\mathbf{X}^{\mathrm{One}})-V_{\theta}(\mathbf{X}^{\mathrm{One}})] \leq 2\epsilon \min(C,N).
    \] 
    Here, the first is bounded by $2\epsilon \min(C,N)$, while the second term is non-positive due to the optimality of $\mathbf{X}^{\mathrm{One}}$ on $\hat{\theta}$.
    For the third term, recall that $\mathbf{X}^{\mathrm{One}}$ depends only on the observed $\hat{\theta}$ and is therefore fixed, so $V_{\hat{\theta}}(\mathbf{X}^{\mathrm{One}})$ is deterministic. In $\mathbf{X}^{\mathrm{One}}$, each patient $i$ who is offered a single provider $j(i)$ realizes utility $\max\{\theta_{ij(i)},\,\theta_{i,M+1}\}$, a convex function of $\theta$. By Jensen's inequality together with $\mathbb{E}[\theta]=\hat{\theta}$,
\[
\mathbb{E}\big[V_{\theta}(\mathbf{X}^{\mathrm{One}})\big]
= \sum_{i=1}^N \mathbb{E}\big[\max\{\theta_{ij(i)},\,\theta_{i,M+1}\}\big]
\;\ge\; \sum_{i=1}^N \max\{\hat{\theta}_{ij(i)},\,\hat{\theta}_{i,M+1}\}
= V_{\hat{\theta}}(\mathbf{X}^{\mathrm{One}}),
\]
so the third term $\mathbb{E}[V_{\hat{\theta}}(\mathbf{X}^{\mathrm{One}}) - V_{\theta}(\mathbf{X}^{\mathrm{One}})] \le 0$, with equality when no patient is near the exit threshold (as in the instance constructed above). Combining the three terms yields $\mathbb{E}[\mathrm{OPT} - V_{\theta}(\mathbf{X}^{\mathrm{One}})] \le 2 \epsilon \min(C,N)$.

\end{proof}

\begin{proof}[Proof of Proposition~\ref{prop:offer_all_poor}]
    Let $\epsilon = 0$ and $N = M = k$, with unit capacities for all $j$ and exit utility 0. Let $\boldsymbol{\sigma}$ be the uniform random ordering. Fix a single pair $(i^*, j^*)$ and set $\theta_{i^*,j^*} = 1$, $\theta_{ij^*} = 2\gamma$ for all $i \neq i^*$, and $\theta_{ij} = \gamma$ for all other $(i,j)$ pairs, with $\gamma \ll 1$. Every patient therefore strictly prefers provider $j^*$, so under offer-all all $N$ patients contend for it; since $\bar{c}_{j^*} = 1$, provider $j^*$ goes to the first arrival.
    The optimal policy achieves a total utility of $1+\gamma(N-1)$, while the offer-all policy achieves a utility of at most $\frac{1}{N} + \gamma (N+1)$ because $i^{*}$ and $j^{*}$ are matched with probability $\frac{1}{N}$ (which is the probability patient $i^{*}$ arrives first).
    The ratio of Offer-All to OPT welfare is at most $\frac{1/N+\gamma(N+1)}{1+\gamma(N-1)}$, which tends to $1/N$ as $\gamma \rightarrow 0$. 
    For any target $\alpha \in (0,1)$ it suffices to take $N = M > 1/\alpha$ and tune $\gamma$ so that the ratio is at most $\alpha$. The guarantee thus degrades without bound as the market grows, since a single unit-capacity provider must serve a growing population.
\end{proof}

\section{SAM Proofs}
\label{app:sam_proofs}
\begin{proof}[Proof of Lemma~\ref{lem:menu_loss}]
Let $y^t$ attain
\[
\max_{\substack{y\text{ feasible}\\ \operatorname{supp}(y)\subseteq E_t(X)}} \sum_{ij}y_{ij}(\mu_i^*+\lambda_j^*).
\]
By complementary slackness for an optimal primal--dual pair $(y^*,\mu^*,\lambda^*)$ of the unrestricted assignment LP,
\[
OPT=\sum_{ij}y_{ij}^*(\mu_i^*+\lambda_j^*).
\]
Since the support of $y^t$ lies in $E_t(X)$, $y^t$ is feasible for the menu-restricted LP.
Moreover, every edge in its support satisfies
\[
\mu_i^*+\lambda_j^*-\Delta_{ij}\le t,
\]
or equivalently
\[
\Delta_{ij}\ge \mu_i^*+\lambda_j^*-t.
\]
Therefore,
\[
LP_{\theta}(X)\ge \sum_{ij}y_{ij}^t\Delta_{ij}\ge \sum_{ij}y_{ij}^t(\mu_i^*+\lambda_j^*-t).
\]
By the definition of $\Def_t(X)$,
\[
\sum_{ij}y_{ij}^t(\mu_i^*+\lambda_j^*)=OPT-\Def_t(X).
\]
Since $y^t$ is capacity feasible,
\[
\sum_{ij}y_{ij}^t\le \sum_j\bar c_j=C.
\]
Hence,
\[
LP_{\theta}(X)\ge OPT-\Def_t(X)-tC.
\]
Rearranging gives
\[
OPT-LP_{\theta}(X)\le tC+\Def_t(X).
\]
\end{proof}

\begin{proof}[Proof of Lemma~\ref{lem:coordination_loss}]
By strong duality for the menu-restricted LP,
\[
LP_{\theta}(X)=\sum_i\mu_i^X+\sum_j\lambda_j^X\bar c_j.
\]
For each matched patient $i$, dual feasibility implies
\[
g^X_{i,s_i}:=\mu_i^X+\lambda_{s_i}^X-\Delta_{i,s_i}\ge0,
\]
so
\[
\Delta_{i,s_i}=\mu_i^X+\lambda_{s_i}^X-g^X_{i,s_i}.
\]
Therefore, the realized welfare can be written as
\[
V(X,\sigma)=\sum_{i\ \mathrm{matched}}\mu_i^X+\sum_j\lambda_j^X f_j-\sum_{i\ \mathrm{matched}}g^X_{i,s_i}.
\]
Subtracting this expression from the dual objective gives
\[
LP_{\theta}(X)-V(X,\sigma)=\sum_{i\ \mathrm{unmatched}}\mu_i^X+\sum_j\lambda_j^X(\bar c_j-f_j)+\sum_{i\ \mathrm{matched}}g^X_{i,s_i}.
\]
Each term is nonnegative because $\mu_i^X\ge0$, $\lambda_j^X\ge0$, $f_j\le\bar c_j$, and $g^X_{i,s_i}\ge0$ by dual feasibility.
\end{proof}

\begin{proof}[Proof of Proposition~\ref{prop:slack}]
For part (i), recall that
\[
r_{ij}^+=(\Delta_{ij}-\lambda_j^*)^+.
\]
Since both $\theta^{(s)}$ and $\theta$ lie within $\epsilon$ of $\hat{\theta}$,
\[
|\Delta_{ij}^{(s)}-\Delta_{ij}|\le4\epsilon.
\]
By the dual stability condition,
\[
|\lambda_j^{\sam}-\lambda_j^*|\le L_\lambda\epsilon.
\]
Since the map $x\mapsto x^+$ is $1$-Lipschitz, for every scenario $s$,
\[
\left|(\Delta_{ij}^{(s)}-\lambda_j^{\sam})^+-(\Delta_{ij}-\lambda_j^*)^+\right|
\le(4+L_\lambda)\epsilon
=\eta\epsilon.
\]
Averaging over scenarios yields
\[
|m_{ij}-r_{ij}^+|\le\eta\epsilon.
\]

For part (ii), by dual optimality,
\[
\mu_i^*
=
\max\left\{0,\max_{\ell}(\Delta_{i\ell}-\lambda_\ell^*)\right\}
=
\max_{\ell}r_{i\ell}^+.
\]
Let $\ell$ attain this maximum, so $r_{i\ell}^+=\mu_i^*$.
Since $j=j^{\sam}(i)$ is designated,
\[
m_{ij}\geq m_{i\ell}.
\]
Applying part (i) gives
\[
r^+_{ij}
\geq
m_{ij}-\eta\eps
\geq
m_{i\ell}-\eta\eps
\geq
r^+_{i\ell}-2\eta\eps
=
\mu_i^*-2\eta\eps.
\]

Moreover, since $j$ is designated, $m_{ij}>0$.
For every scenario $s$,
\[
\Delta_{ij}^{(s)}-\lambda_j^{\sam}
=
r_{ij}
+
(\Delta_{ij}^{(s)}-\Delta_{ij})
-
(\lambda_j^{\sam}-\lambda_j^*)
\leq
r_{ij}+\eta\eps.
\]
Thus, if $r_{ij}<-\eta\eps$, then
$\Delta_{ij}^{(s)}-\lambda_j^{\sam}<0$ for every $s$, which would imply
$m_{ij}=0$, a contradiction.
Therefore,
\[
r_{ij}\geq-\eta\eps.
\]

If $\mu_i^*>2\eta\epsilon$, then
\[
r_{ij}^+\geq\mu_i^*-2\eta\epsilon>0,
\]
so $r_{ij}=r_{ij}^+$ and
\[
g_{ij}
=
\mu_i^*-r_{ij}
\leq
2\eta\epsilon.
\]

Otherwise, $\mu_i^*\leq2\eta\epsilon$.
Using $r_{ij}\geq-\eta\epsilon$,
\[
g_{ij}
=
\mu_i^*-r_{ij}
\leq
2\eta\epsilon+\eta\epsilon
=
3\eta\epsilon.
\]
Thus every designated edge satisfies
\[
g_{ij}\leq3\eta\epsilon.
\]
\end{proof}

\begin{proof}[Proof of Proposition~\ref{prop:sam_menu_loss}]
Let $I^*:=\{i:\mu_i^*>0\}$, and assign each $i\in I^*$ to their
designated provider $j^{\sam}(i)$. By
the balanced demand condition, this allocation is feasible
and therefore provides a lower bound on $LP_{\theta}(X^{\sam})$.

By Proposition~\ref{prop:slack}, every designated edge has dual slack
at most $3\eta\epsilon$. Thus,
\[
LP_{\theta}(X^{\sam})
\ge
\sum_{i\in I^*}\mu_i^*
+
\sum_j \lambda_j^*D_j
-
3\eta\epsilon |I^*|.
\]
Next, we note that $|I^*|\le C$ because $|I^*| = \sum_{i \in I^*} \sum_{j} y^{*}_{ij} \leq \sum \bar{c}_{j} = C$. 
So
\[
OPT=\sum_{i\in I^*}\mu_i^*+\sum_j\lambda_j^*\bar c_j,
\]
and we obtain
\[
OPT-LP_{\theta}(X^{\sam})
\le
3\eta\epsilon C
+
\sum_j\lambda_j^*(\bar c_j-D_j).
\]
The balanced demand condition shows that the second term is zero, showing that $OPT-LP_{\theta}(X^{\sam})
\le
3\eta\epsilon C$. 
\end{proof}

\begin{proof}[Proof of Proposition~\ref{prop:sam_coordination}]
Fix any response order $\sigma$. For each provider $j$, at most $T_j$ patients can select $j$. By the slack-supported alternatives condition,
\[
T_j\leq \bar c_j,
\]
so no provider can become unavailable to a patient for whom it is the designated provider.

Therefore, when each $i\in I^*$ responds, their designated provider $j^{\sam}(i)$ remains available. 
Since patients select their highest-utility available option, patient $i$ obtains utility at least $\Delta_{ij^{\sam}(i)}$.
Summing over $i\in I^*$ gives
\[
V(X^{\sam},\sigma)
\geq
\sum_{i\in I^*}\Delta_{ij^{\sam}(i)}.
\]
We note that $\sum_{i\in I^*}\Delta_{ij^{\sam}(i)}$ is within $O(\epsilon C)$ of $\opt$ by the proof of Proposition~\ref{prop:sam_menu_loss}. 
Hence
\[
OPT-V(X^{\sam},\sigma) \leq O(\eps C).
\]
Combining this with Proposition~\ref{prop:sam_menu_loss} yields
\[
OPT - V(\mathbf{X}^{\sam},\sigma)
\leq
3\eta\epsilon C
=
O(\epsilon C).
\]
\end{proof}
\end{document}